\documentclass[a4paper,USenglish,cleveref,colorlinks=true,urlcolor=black,linkcolor=black,citecolor=black,thm-restate,authorcolumns]{lipics-v2021}

\nolinenumbers
\hideLIPIcs

\usepackage{amsfonts}
\usepackage{subcaption}
\usepackage[table]{xcolor}
\usepackage{csquotes}
\usepackage{booktabs}
\usepackage{comment}
\usepackage{microtype}
\usepackage{mathtools}
\usepackage{bm}

 \usepackage{thm-restate}
 \usepackage{amssymb}
\PassOptionsToPackage{capitalise}{cleveref}
 \crefname{section}{Section}{Sections}
\PassOptionsToPackage{colorlinks=true,urlcolor=black,linkcolor=black,citecolor=black}{hyperref}

\newcommand{\prooflink}[1]{\hypersetup{linkcolor=lipicsBulletGray}\hyperref[#1]{$\blacktriangledown$}}
\newcommand{\statlink}[1]{\hypersetup{linkcolor=lipicsYellow}\hyperref[#1]{$\blacktriangle$}}
\newcommand{\mylink}{}
\newcommand{\cprooflink}[1]{\hypersetup{linkcolor=lipicsBulletGray}\hyperref[#1]{$\triangledown$}}
\newcommand{\cstatlink}[1]{\hypersetup{linkcolor=lipicsYellow}\hyperref[#1]{$\vartriangle$}}
\newcommand{\cmylink}{}

\theoremstyle{claimstyle}

\newcommand{\definitionBox}[1]{\smallskip
\noindent\fbox{\begin{minipage}{0.98\textwidth}
\smallskip

#1

\smallskip
\end{minipage}}

\smallskip}

\def\defn#1{\textit{\textbf{\boldmath #1}}}
\renewcommand{\emph}[1]{\defn{#1}}

\title{Rerouting Curves on Surfaces}

\author{Timo Brand}{Technische Universität München, DE}{timo.brand@tum.de}{https://orcid.org/0009-0004-3111-2045}{%
}
\author{Stefan Felsner}{Technische Universität Berlin, DE}{felsner@math.tu-berlin.de}{https://orcid.org/0000-0002-6150-1998}{%
}
\author{Henry Förster}{John Cabot University, IT \and Technische Universität München, DE \and Universität Tübingen, DE}{henry.foerster@johncabot.edu}{https://orcid.org/0000-0002-1441-4189}{%
}
\author{Stephen Kobourov}{University of Arizona, US \and Technische Universität München, DE}{stephen.kobourov@tum.de}{https://orcid.org/0000-0002-0477-2724}{%
}
\author{Anna Lubiw}{University of Waterloo, CA}{alubiw@uwaterloo.ca}{https://orcid.org/0000-0002-2338-361X}{%
}
\author{Yoshio Okamoto}{The University of Electro-Communications, JP}{okamotoy@uec.ac.jp}{https://orcid.org/0000-0002-9826-7074}{Partially supported by JSPS KAKENHI Grant Numbers JP23K10982, JP26K23806, and JST ERATO Grant Number JPMJER2301.
}
\author{János Pach}{Rényi Institute of Mathematics, HU \and EPFL, CH}{pach@cims.nyu.edu}{https://orcid.org/0000-0002-2389-2035}{%
}
\author{Csaba D. Tóth}{Cal State Northridge, Los Angeles, CA, US \and Tufts University, Medford, MA, US}{csaba.toth@csun.edu}{https://orcid.org/0000-0002-8769-3190}{Research supported in part by the NSF award DMS-2154347.}
\author{G\'eza Tóth}{Rényi Institute of Mathematics, HU}{geza@renyi.hu}{https://orcid.org/0000-0003-1751-6911}{%
}
\author{Torsten Ueckerdt}{Karlsruhe Institute of Technology, DE}{torsten.ueckerdt@kit.edu}{https://orcid.org/0000-0002-0645-9715}{Funded by the Deutsche Forschungsgemeinschaft (DFG, German Research Foundation) – 520708409}
\author{Pavel Valtr}{Charles University Prague, CZ}{valtr@kam.mff.cuni.cz}{https://orcid.org/0000-0002-3102-4166}{%
}

\authorrunning{Brand, Felsner, Förster, Kobourov, Lubiw, Okamoto, Pach, Tóth, Tóth, Ueckerdt, Valtr}

\Copyright{Timo Brand, Stefan Felsner, Henry Förster, Stephen Kobourov, Anna Lubiw, Yoshio Okamoto, János Pach, Csaba D.\ Tóth, Géza Tóth, Torsten Ueckerdt, and Pavel Valtr}

\ccsdesc[500]{Mathematics of computing~Geometric topology}
\ccsdesc[500]{Mathematics of computing~Graph algorithms}
\ccsdesc[500]{Mathematics of computing~Trees}
\ccsdesc[100]{Mathematics of computing~Combinatorial algorithms}

\keywords{Combinatorial reconfiguration, orientable surface, non-orientable surface, rerouting}

\acknowledgements{This work started at Dagstuhl Seminar 24062
``Beyond-Planar Graphs: Models, Structures and Geometric Representations''.}

\EventEditors{Philip Bille, Seth Pettie, and Sabine Storandt}
\EventNoEds{3}
\EventLongTitle{34th Annual European Symposium on Algorithms (ESA 2026)}
\EventShortTitle{ESA 2026}
\EventAcronym{ESA}
\EventYear{2026}
\EventDate{August 31--September 4, 2026}
\EventLocation{L'Aquila, Italy}
\EventLogo{}
\SeriesVolume{388}
\ArticleNo{26}
\begin{document}
\maketitle

\begin{abstract}
    We study the problem of reconfiguring a crossing-free embedding of a graph on a surface, with edges represented as curves, into another crossing-free embedding of the same graph on the same surface with the same fixed vertex positions.
    In this process, we reroute one edge at a time while maintaining crossing-free intermediate embeddings.
    This problem was introduced by Ito et al. [TALG 2025], who showed that even if the graph is a matching of two edges, reconfiguration is not always possible in the plane, but is always possible on the torus.
    For matchings of two or more edges, they gave a necessary and sufficient condition for reconfigurable embeddings in the plane, but not on the torus.
    Our main result is that for matchings, trees and forests, reconfiguration is always possible on the torus, and consequently, on any orientable surface of genus at least one.
    In addition, we provide sufficient conditions for reconfiguration on orientable surfaces of genus at least one and in the projective plane.
    For more general graphs, we show that reconfiguration is not always possible.
\end{abstract}

\begin{figure}[htb!]
    \centering
    \begin{subfigure}{0.2\textwidth}
        \centering
        \includegraphics[scale=0.85,page=2]{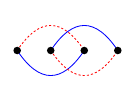}
        \subcaption{}
        \label{fig:intro:1}
    \end{subfigure}
    \hfil\begin{subfigure}{0.75\textwidth}
        \centering
        \includegraphics[scale=0.85,page=3]{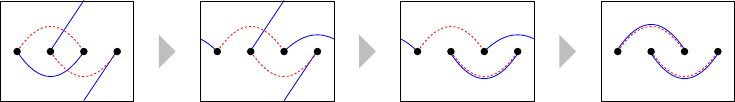}
        \subcaption{}
        \label{fig:intro:2}
    \end{subfigure}

    \begin{subfigure}{\textwidth}
        \centering
        \includegraphics[scale=0.85,page=4]{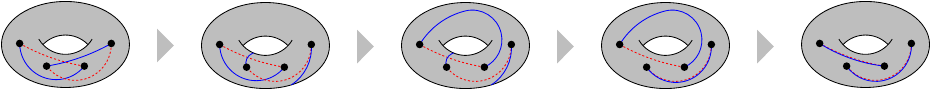}
        \subcaption{}
        \label{fig:intro:3}
    \end{subfigure}
    \caption{(a)~Two embeddings (red \& blue) that cannot be reconfigured on the plane.
        (b)~Reconfiguration on the fundamental square.
        (c)~Corresponding illustration on the torus.}
    \label{fig:intro}
\end{figure}

\section{Introduction}
\label{sec:intro}
We study the problem of reconfiguring graph embeddings on a surface, where the vertices are fixed and each reconfiguration step redraws one edge (represented as a curve).
Consider a set~$P$ of points on a surface $\Sigma$, and two embeddings $\mathcal{B}$ and $\mathcal{R}$ of the same simple graph $G$
that share the same mapping of vertices to $P$.
Here, an \defn{embedding} means that each edge is drawn as a simple curve, which we call an \emph{edge curve}, on the surface, and no two edge curves intersect except at a common endpoint.
The edge curves of $\mathcal{B}$ may cross the edge curves of $\mathcal{R}$ and the correspondence between edge curves of $\mathcal{B}$ and $\mathcal{R}$ is implied by the identical mapping of vertices.
A \emph{reconfiguration step} or \emph{move} replaces one edge curve~$\gamma$ of an embedded graph $G$ by a new curve $\gamma'$ to obtain a new embedding of $G$; i.e., $\gamma'$ may not cross any of the other edge curves in the embedding, though we allow $\gamma$ and $\gamma'$ to intersect.
We address the question  whether $\mathcal{B}$ can be reconfigured to~$\mathcal{R}$ via a sequence of moves.

Ito et al.~\cite{ItoIK0MNOO25} showed
that reconfiguration in the plane is not always possible, even if the graph is a matching of two
edges; see \cref{fig:intro:1}.
On the other hand, they proved that reconfiguration of a matching of two edges is always possible on a surface $\Sigma$ of genus $g \ge 1$;
see \cref{fig:intro:2,fig:intro:3} for
an example on the torus.
In this paper, we investigate the natural open problem of determining
which classes of graphs can be reconfigured on surfaces of higher genus.
In particular, we investigate general forests (including larger matchings).

\subparagraph*{Related Work.}
Rerouting curves on surfaces was introduced by Ito et al.\ \cite{ItoIK0MNOO25}.
They studied the case where the graph is a matching (a set of independent edges) and the surface is the plane (equivalently the sphere), for which they gave an algebraic condition characterizing the reconfigurable embeddings.
This was a main step for studying the reconfigurability of vertex-disjoint paths in planar graphs.
When the graph is a matching of two edges on a torus, they showed that reconfiguration is always possible, and left open the case of matchings of three or more edges.
The problem of reconfiguring graph embeddings on a surface using elementary moves (studied in this paper) has obvious similarities to other problems in graph drawing and computational topology.
We highlight key differences that may explain why the techniques developed for other problems do not seem to be helpful for our problem.
The problem of \emph{morphing graph drawings on a surface}~\cite{alamdari2017morph,chambers2021morph,ericksonplanar} is different in that the vertices are allowed to move, and usually the edges must remain straight segments on the surface.
The problem of \emph{tightening} or \emph{untangling curves on a surface}~\cite{DBLP:journals/talg/ChangM22,DBLP:journals/dcg/ChangE17,DBLP:conf/soda/Chang0LMSSTT18,de2024untangling} is also different in that they consider drawings with possible crossings (i.e., immersions rather than embeddings), and deform the edge curves continuously via so-called homotopy moves (local moves that modify the topology of the immersion).
We also point the interested reader to Colin de Verdi\`ere's
survey~\cite{de2017computational} on graphs on surfaces.
Finally, our problem is a type of \emph{combinatorial reconfiguration} problem; see~\cite{DBLP:journals/algorithms/Nishimura18,DBLP:books/cu/p/Heuvel13} for introductions to this field.
Note that a reconfiguration of circle arrangements on orientable surfaces with moves that replace one circle with a transversely intersecting circle has been considered in~\cite{hatcher1980presentation,lackenby2024bounds}.

An overview of our results and open problems is presented in \cref{tab:results}.
Due to space constraints, some (details of) proofs are provided in the appendix only; their statements are marked with \textcolor{yellow}{$\blacktriangle$}\textcolor{darkgray}{$\blacktriangledown$}.
In the PDF, \textcolor{yellow}{$\blacktriangle$} links to the statement and \textcolor{darkgray}{$\blacktriangledown$} links to the proof.

\begin{table}[t]
    \caption{Our reconfiguration results and our open problems (numbered as in the list in \cref{sec:open}).}
    \label{tab:results}
    \centering
    \sffamily \setlength{\tabcolsep}{6pt}
    \renewcommand{\arraystretch}{1.5}
    \small
    \begin{tabular}{llcc}
          \rowcolor{lipicsBulletGray!30!black}\multicolumn{2}{l}{\bfseries \textcolor{lipicsYellow!5}{GRAPH CLASS}} & \multicolumn{1}{c}{\bfseries \textcolor{lipicsYellow!5}{ORIENTABLE}}          & \multicolumn{1}{c}{\bfseries \textcolor{lipicsYellow!5}{NON-ORIENTABLE}}                                                                                      \\
        {\cellcolor{lipicsBulletGray!35}\bfseries Forests}                       &\cellcolor{lipicsBulletGray!35}\textit{-- General Case}                          &  \cellcolor{lipicsBulletGray!35}                                                                            &         \cellcolor{lipicsBulletGray!25}                                           \\\cellcolor{lipicsBulletGray!35}
                                                  & \cellcolor{lipicsBulletGray!35}\textit{-- Perfect Matchings (arbitrary)}         &      \cellcolor{lipicsBulletGray!35}                                                                        & \cellcolor{lipicsBulletGray!25}\multirow{-2}{*}{Open Problem~\ref{openproblem:1}} \\ \cellcolor{lipicsBulletGray!35}
        &\cellcolor{lipicsBulletGray!35} \textit{-- Perfect Matchings (inside disk)}       & \multirow{-3}{*}{\cellcolor{lipicsBulletGray!35}\cref{thm:tree}}                                            & \cellcolor{lipicsBulletGray!35}\cref{thm:projective}                              \\
        \cellcolor{lipicsYellow!20}{\bfseries Planar graphs} \phantom{AH}    &\cellcolor{lipicsYellow!20}\textit{-- General Case}                          & \multicolumn{2}{c}{\cellcolor{lipicsYellow!5}Open Problems~\ref{openproblem:2} \&~\ref{openproblem:3}}                                                      \\ \cellcolor{lipicsYellow!20}
                                                  & \cellcolor{lipicsYellow!20}\textit{-- with Fixed Rotation System}            & \cellcolor{lipicsYellow!20}\cref{thm:embedding-preserving}                                              &    \cellcolor{lipicsYellow!5}                                                \\ \cellcolor{lipicsYellow!20}
                                 &\cellcolor{lipicsYellow!20}\textit{-- Series-parallel graphs}                &\cellcolor{lipicsYellow!20} \cref{thm:seriesparallel}                                                    & \cellcolor{lipicsYellow!5} \multirow{-2}{*}{Open Problem~\ref{openproblem:3}} \\
        \cellcolor{lipicsBulletGray!35}{\bfseries Negative examples}

                                                  & \cellcolor{lipicsBulletGray!35}\textit{-- Same rotation system}                  &    \cellcolor{lipicsBulletGray!35}             & \cellcolor{lipicsBulletGray!35}                                                            \\
                                  \cellcolor{lipicsBulletGray!35}                & \cellcolor{lipicsBulletGray!35}\textit{-- Different rotation systems,}           &            \multicolumn{2}{c}{\cellcolor{lipicsBulletGray!35}{\cref{thm:negative}}}                                                                                                                          \\\cellcolor{lipicsBulletGray!35}
                                 \cellcolor{lipicsBulletGray!35}                 &\cellcolor{lipicsBulletGray!35}\textit{\phantom{--} all implementing embeddings}           &\cellcolor{lipicsBulletGray!35}                                         &\cellcolor{lipicsBulletGray!35}                                                                                  \\
        \end{tabular}
\end{table}

\section{Preliminaries}
\label{sec:pre}

We study the reconfiguration of graphs embedded on a surface $\Sigma$ where edges are represented as arcs.
A curve on a surface $\Sigma$ is a continuous function $\gamma:[0,1]\to\Sigma$.
It is \emph{closed} if $\gamma(0)=\gamma(1)$ and \emph{simple} if $\gamma$ is injective on $[0,1)$.
An \emph{arc} on $\Sigma$ is the image of a simple curve $\gamma:[0,1]\to\Sigma$ such that $\gamma(0)\neq\gamma(1)$.
The points $\gamma(0)$ and $\gamma(1)$ are called the endpoints of the arc.
We will always consider simple curves unless mentioned otherwise.

Let $G=(V,E)$ be a simple graph with $n$ vertices
$V=\{v_1,\ldots,v_n\}$ and with $m$ edges $E\subseteq \binom{V}{2}$.
Let $\Sigma$ be a surface and let $P=\{p_1,\ldots,p_n\} \subset \Sigma$ be a set of $n$ distinct points on $\Sigma$.
A mapping $\mathcal{E}$ defined on $V\cup E$ is a \emph{$\Sigma$-embedding} of~$G$~on~$P$~if
\textbf{1.} $\mathcal{E}(v_i) = p_i$ for $i \in \{1,\dots,n\}$,
\textbf{2.} $\mathcal{E}(e)$ is a simple open arc on $\Sigma$ with endpoints $\mathcal{E}(u)$ and $\mathcal{E}(v)$, called an \defn{edge curve}, for $e = uv \in E$,
\textbf{3.} $\mathcal{E}(e) \cap \mathcal{E}(f) \subseteq \mathcal{E}(V)$ for $e,f\in E$, $e\neq f$.
For simplicity, we use $v_i$ and  $p_i$ interchangeably and denote by $\mathcal{E}$  both the mapping and its image on $\Sigma$.
Let~$G'$ be a proper subgraph of a graph $G$.
A $\Sigma$-embedding of $G'$ is called a \defn{partial} $\Sigma$-embedding of $G$; that is, it describes an embedding of only a subset of the edges of $G$. Given a $\Sigma$-embedding $\mathcal{E}$ of $G$, we denote by $\mathcal{E}[G']$ the restriction $\mathcal{E}$ to the subgraph $G'$, i.e.,  $\mathcal{E}[G']$ contains exactly the edge curves of $\mathcal{E}$ corresponding to edges in $G'$. The \emph{rotation system} of an embedding $\mathcal{E}$ specifies the clockwise cyclic order of edges incident to each vertex.

We investigate embeddings of a graph $G$ where the mapping of vertices to points is fixed,
that is, we consider $\Sigma$-embeddings $\mathcal{E}$ and $\mathcal{F}$, where $\mathcal{E}(v_i)=\mathcal{F}(v_i)=p_i$.
Two such embeddings are \emph{adjacent} if there is exactly one edge $e^* \in E$ such that $\mathcal{E}(e^*) \neq \mathcal{F}(e^*)$.
Embedding $\mathcal{E}$ is \emph{reconfigurable} to $\mathcal{F}$ if there is a \emph{reconfiguration sequence}
$(\mathcal{E}_0 = \mathcal{E},\mathcal{E}_1,\ldots,\mathcal{E}_k = \mathcal{F})$ from $\mathcal{E}$ to $\mathcal{F}$ where each $\mathcal{E}_i$ is a $\Sigma$-embedding of $G$ on $P$ for $i \in \{0,\dots,k\}$,
and $\mathcal{E}_i$ is adjacent to  $\mathcal{E}_{i+1}$ for $i \in \{0,\dots,k-1\}$.
We study the following \texttt{$\Sigma$-Embedding Reconfiguration} problem:

\definitionBox{\hspace{5pt}\texttt{$\Sigma$-Embedding Reconfiguration}

    \hspace{5pt}\textbf{Input.} A graph $G=(V,E)$ and two embeddings $\mathcal{B}$ and $\mathcal{R}$ of $G$ on surface $\Sigma$.

    \hspace{5pt}\textbf{Question.} Is $\mathcal{B}$ reconfigurable to $\mathcal{R}$?}

Only edges are reconfigured and relevant for crossings; hence assume that $G$ has~no isolated vertices.
Ito et al.~\cite{ItoIK0MNOO25} solved the
problem when $G$ is a matching and $\Sigma$ is the plane by giving necessary and sufficient conditions for reconfiguration.
Their characterization yields a polynomial-time algorithm
in special cases.
Their example of two disjoint edges (Figure~\ref{fig:intro}) is the smallest example where reconfiguration is not possible in the plane.
However, they show how to reconfigure embeddings of this 2-edge graph on any surface of higher genus.

We denote by $\Sigma$ a compact surface without boundary.
Such a surface $\Sigma$ is an \emph{orientable surface of genus $0$} if it is homeomorphic to the sphere.
An \emph{orientable surface of genus $g$} is obtained from a surface $\Sigma'$ of genus $g-1$ by removing two disjoint open disks from  $\Sigma'$,
and identifying their boundaries with opposite orientations.
The {orientable surface of genus $1$} is the \emph{torus}.
Embeddings on the sphere and torus are called
\emph{plane} and \emph{toroidal}, respectively.
There are also \emph{non-orientable surfaces}.
The \emph{non-orientable surface of genus $1$} is the projective plane $\mathbb{R}P^2$.
It is obtained from a sphere $\mathbb{S}^2$ by removing the interior of a disk $D$ and identifying antipodal points on the boundary $\partial D$; this operation is called \emph{attaching a crosscap}.
We refer to the image of $\partial D$ in the quotient as the \emph{crosscap}.
We say that an $\mathbb{R}P^2$-embedding of a graph $G$ has a \emph{clean crosscap} if the embedding of every edge of $G$ is disjoint from the crosscap.
For $g \ge 2$, the \emph{non-orientable surface of genus $g$} is obtained from the non-orientable surface of genus $g-1$ by attaching a crosscap.

A closed curve $\gamma$ on a surface $\Sigma$ is \emph{separating} if $\Sigma \setminus \gamma$ is disconnected;
otherwise~$\gamma$ is \emph{non-separating}.
The curve $\gamma$ is \emph{contractible} if it is homotopic to a point.
A \emph{disk}~on~a surface~is the region bounded by a contractible closed curve $\gamma$.
On the sphere, each closed curve is contractible and separating;
on the torus, all separating closed curves are contractible~\cite{viro2024elementary}.

The \emph{fundamental square} $\bar{\Sigma}$
is the following quotient-space representation of the torus $\Sigma=\mathbb{S}^1\times \mathbb{S}^1$: it is the quotient of a unit square $[0,1]^2$, where
\textbf{1.} the left and right sides are identified such that $(0,y)$ is identified with $(1,y)$ for all $y\in [0,1]$, and
\textbf{2.} the top and bottom sides are identified such that $(x,0)$ is identified with $(x,1)$ for all $x\in [0,1]$.
We also call the boundary (interior) of the unit square the \emph{boundary} (\emph{interior}) of the fundamental square, and denote it $\partial \bar{\Sigma} =\partial [0,1]^2$ (resp., ${\rm int}\,\bar{\Sigma} = {\rm int}\,[0,1]^2$).
Note that the boundary (interior) of the fundamental square is not the boundary (interior) in a topological sense; instead it is associated with the particular representation of the torus as a fundamental square.
If $\mathcal{B}$ is an embedding on the torus $\Sigma$, then its image under this representation is called a \emph{$\bar{\Sigma}$-embedding}.
More generally, if $\Sigma$ is an orientable surface of genus $g\geq 1$, then $\Sigma$ is homeomorphic to the quotient of a $4g$-gon, called a \emph{fundamental polygon}, whose sides are identified according to the standard scheme
$a_1b_1a_1^{-1}b_1^{-1}\cdots a_gb_ga_g^{-1}b_g^{-1}$.
For $g=1$, this is exactly the square model of the torus.
If no edge curve of $\mathcal{B}$ intersects $\partial\bar{\Sigma}$, equivalently if $\mathcal{B}$ is contained in ${\rm int}\,\bar{\Sigma}$, then we call $\mathcal{B}$ an \emph{interior $\bar{\Sigma}$-embedding}.
In this case, $\mathcal{B}$ is contained in a topological disk and can be viewed as a planar embedding inside the chosen fundamental polygon.

\section{Reconfiguring Embeddings of a Forest on the Torus}
\label{sec:torus}

We describe an algorithm that constructs a reconfiguration sequence from a given source $\Sigma$-embedding $\mathcal{B}$ of a forest $F=(V,E)$ to a given target $\Sigma$-embedding $\mathcal{R}$ of $F$ where $\Sigma$ is an orientable surface of genus $g \geq 1$.
For now, let  $\Sigma$ be a torus; later, we generalize to $g>1$.

To formalize our result algorithmically and analyze the performance of our algorithm, we consider a \emph{geometric setting} with the following input specifications:
\begin{enumerate}[1.]
    \item $\mathcal{B}$ and $\mathcal{R}$ are given as $\bar{\Sigma}$-embeddings of $\Sigma$,
    \item no vertex is on the boundary $\partial\bar\Sigma$,
    \item each edge curve is a polyline of straight-line segments, and
    \item each segment crosses $\partial \bar{\Sigma}$ at most once.
\end{enumerate}
For an embedding $\mathcal E$, let $s({\mathcal E})$ be its number of segments, and let $k({\mathcal E})$
be the number of these segments that cross $\partial \bar{\Sigma}$.
The length of the reconfiguration sequence and the runtime will depend on
the parameters $s$ and $k$  of $\mathcal B$ and $\mathcal R$.
Note that although $\cal B$ and $\cal R$ (blue and red in all figures) are each crossing-free $\bar{\Sigma}$-embeddings on ${\Sigma}$, $\cal B$ and $\cal R$ may cross each other.
Our  main result is:

\begin{restatable}{rtheorem}{thmforestontorus}
    \label{thm:tree}
    \renewcommand{\mylink}{\statlink{thm:tree}\prooflink{Pthmforestontorus}}
    Let $F = (V,E)$ be a forest having $c=c(F)$ connected components. Let $\mathcal{B}$ and $\mathcal{R}$ be two
    $\bar{\Sigma}$-embeddings of $F$ on the fundamental square ${\bar\Sigma}$.
    Suppose that the edges of $F$ are represented by polylines in  $\mathcal{B}$ and $\mathcal{R}$ and let $s = s(\mathcal{B}) + s(\mathcal{R}) $ be the total number of segments in~$\mathcal{B}$ and~$\mathcal{R}$ and let $k = k(\mathcal{B}) + k(\mathcal{R})$ be the total number of segments in~$\mathcal{B}$ and~$\mathcal{R}$  that intersect the boundary~$\partial\bar\Sigma$ of~$\bar\Sigma$.
    Then there is a reconfiguration sequence from~$\mathcal{B}$ to~$\mathcal{R}$ of length~$O(c 3^k s^2)$.
    Each embedded graph in the sequence has~$O(c^2 3^{2k} s^3)$ segments.
    There is an algorithm to compute the reconfiguration sequence in time proportional to the output size, where the output size is the total number of segments in an explicit list of all embeddings in the sequence.
\end{restatable}

\noindent Our algorithm proving \cref{thm:tree} consists of the following two phases.
\begin{itemize}
    \item \textsf{\bfseries \textcolor{lipicsLineGray}{Phase~1:  Reconfiguring $\cal B$ and $\cal R$ to $\mathcal{B}^*$ and  $\mathcal{R}^*$}:} Separately reconfigure $\cal B$ and $\cal R$ to embeddings $\mathcal{B}^*$ and  $\mathcal{R}^*$, respectively, such that no edge curve in $\mathcal{B}^*$ or  $\mathcal{R}^*$ crosses the boundary  $\partial \bar{\Sigma}$. (We completely ignore $\cal R$ when reconfiguring $\cal B$, and vice-versa.)

    \item \textsf{\bfseries \textcolor{lipicsLineGray}{Phase~2:  Reconfiguring $\mathcal{B}^*$ to  $\mathcal{R}^*$:}}
          Reconfigure
          $\mathcal{B}^*$ to  $\mathcal{R}^*$ using the freedom to temporarily reconfigure edges to wrap around the boundary of $\bar\Sigma$.
\end{itemize}

We obtain the final sequence by
first reconfiguring $\mathcal{B}$ to $\mathcal{B}^*$ (\textsf{\bfseries \textcolor{lipicsLineGray}{Phase~1}}), then $\mathcal{B}^*$ to $\mathcal{R}^*$ (\textsf{\bfseries \textcolor{lipicsLineGray}{Phase~2}}), and finally $\mathcal{R}^*$ to $\mathcal{R}$ (\textsf{\bfseries \textcolor{lipicsLineGray}{Phase 1}}; reconfiguration sequences are reversible).
For \textsf{\bfseries \textcolor{lipicsLineGray}{Phase~1}}, it suffices to show how to reconfigure $\mathcal{B}$ to $\mathcal{B}^*$, ignoring $\mathcal{R}$ ($\mathcal{R}$ to $\mathcal{R}^*$ is symmetric).

\paragraph*{Phase 1. Reconfigure $\mathcal{B}$ to a $\bar\Sigma$-embedding $\mathcal{B}^*$  that does not intersect the boundary}

\begin{restatable}{rlemma}{lemmakingpplane}
    \label{lem:making-p-plane}
    \renewcommand{\mylink}{\statlink{lem:making-p-plane}\prooflink{Plemmakingpplane}}
    Let $\mathcal B$ be a $\bar{\Sigma}$-embedding of a forest $F=(V,E)$ on ${\Sigma}$.
    Let $s = s({\mathcal B})$ and $k = k({\mathcal B})$.
    Then there is a reconfiguration sequence of length $O(3^k)$ from $\mathcal{B}$ to a $\bar{\Sigma}$-embedding $\mathcal{B}^*$ of $F$ that does not cross $\partial \bar{\Sigma}$.
    Each embedded graph in the sequence has $O(3^k s)$ segments.
    There is an algorithm that computes the reconfiguration sequence in time $O(3^{2k} s)$.
\end{restatable}

\begin{proof}

    We first prove the existence of $\mathcal{B}^*$ by induction and then prove the bounds and show how the proof yields an algorithm with the claimed runtime.
    Inspired by the machinery of weak embeddings by Chang, Erickson, and Xu~\cite{ChangEX15}, we use a sequence of embedded trees, ${\mathcal F}_0,{\mathcal F}_1,\dots,{\mathcal F}_k$, that we call the \emph{frame trees}.
    Each frame tree ${\mathcal F}_j$ is a Steiner tree: its vertices include some endpoints of segments of $\mathcal{B}$ and $2k$ additional points
    in $\bar{\Sigma}\setminus \partial\bar{\Sigma}$, at a fixed (small) distance $\delta>0$ from $\partial\bar{\Sigma}$.
    Throughout the reconfiguration sequence, the rerouted edge curves of $\mathcal B$ stay in a small neighborhood of the edges of one of~${\mathcal F}_0,\ldots,{\mathcal F}_k$.
    At the end of the reconfiguration process, they lie in a small neighborhood of the edges of the last frame tree ${\mathcal F}_k$, which the construction below shows does not cross $\partial\bar\Sigma$.
    Consequently, the rerouted edge curves in their final position also do not cross $\partial\bar\Sigma$.
    An edge curve of $\mathcal B$ may spiral around ${\mathcal F}_j$ multiple times.

    \begin{figure}[t!]
        \centering
        \includegraphics[scale=1]{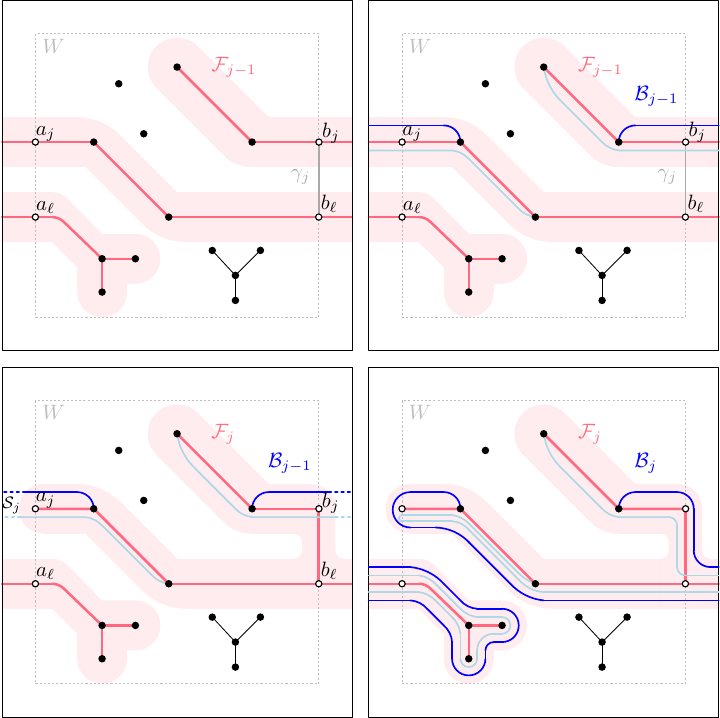}
        \caption{Phase 1:
            \textit{Top-Left:} The smaller square $W$ 
            intersected by the segments crossing $\partial \bar{\Sigma}$ 
            is shown dashed in gray, and frame trees $\mathcal{F}_{j-1}$ and $\mathcal{F}_{j}$  in pink, with their $\varepsilon$-neighborhoods in a lighter color.
            The frame tree $\mathcal{F}_{j-1}$ is cut by deleting the segment $f_j = a_j b_j$
            (eliminating one crossing with $\partial\bar\Sigma$) and then reconnected by adding the candidate curve $\gamma_j$ shown in gray.
            \textit{Top-Right:} In the embedding $\mathcal{B}_{j-1}$, edges may be routed in the $\varepsilon$-neighborhood of the segment $f_j$ (see the cyan and blue edge for two examples). 
            \textit{Bottom-Left:} After updating the frame tree to $F_{j}$, the segments of such edges in the $\varepsilon$-neighborhood of $f_j$ (drawn dashed) must be reconfigured.
            \textit{Bottom-Right:}~The two edges that originally crossed $W$ in the embedding $\mathcal{B}_{j-1}$
            (cyan and blue) are reconfigured, first cyan, then blue, to obtain the embedding $\mathcal{B}_{j}$.\\
            Note that the polylines representing the edges in the embeddings cannot cross themselves and each other, but they can cross the frame tree; in the figure, insertion of the candidate curve $\gamma_j$ to the frame tree is responsible for the intersection of the frame tree with the cyan polyline near $b_j$.
            Moreover, observe that the parts of $\mathcal{F}$ not intersecting $W$ are not in the neighborhood of $\mathcal{F}_{j-1}$ as they are not part of $\mathcal{B}_0$.
        }
        \label{fig:phase1}
    \end{figure}

    If $k({\mathcal B})=0$, then the lemma holds with ${\mathcal B}^*$ set to $\mathcal B$.
    In the following, we assume that $k>0$.
    We fix $\delta>0$ small enough such that all vertices of $\mathcal B$ (including the isolated vertices) lie at a distance larger than $\delta$ from $\partial\bar\Sigma$.
    Thus, all vertices of $\mathcal B$ lie inside the square $W\coloneqq [\delta,1-\delta]^2$.
    Let $\mathcal B_{0}\subseteq \mathcal B$ be the embedded forest obtained from $\mathcal B$ by deleting all connected components that do not cross $\partial\bar\Sigma$.
    Observe that the region $\bar\Sigma\setminus W$ is intersected only by those $k=k({\mathcal B})=k({\mathcal B}_{0})$ segments of $\mathcal B_{0}$ that cross $\partial\bar\Sigma$. Let $s_1,\dots,s_k$ denote these $k$ segments, in arbitrary order.
    Recall that we are in the geometric setting, where, in particular, we assume that in the input embedding each segment crosses the boundary at most once.
    Our goal is to reroute the segments $s_1,\dots,s_k$ to new polylines (one per segment), called their \emph{polyline representations},
    which will all lie in the union of small $\varepsilon$-neighborhoods%
    \footnote{In a slight abuse of notation, we will denote by $\varepsilon$ a suitably small distance without explicitly specifying its value.
    Observe that $\varepsilon$ is not constant:
    Most occurrences of $\varepsilon$ could be numbered consecutively $\varepsilon_1,\varepsilon_2,\ldots$ so that for each $i$, we have $\varepsilon_{i+1} \ge \varepsilon_i$,
    and the last (biggest) one of them is still very small.}
    of the other segments of $\mathcal B_{0}$ and of $\partial W$.
    Each $s_i$ intersects $\partial W$ in two points, $a_i$ and $b_i$.
    We set $A:=\{a_1,\dots,a_k,b_1,\dots,b_k\}$.
    The $2k$ points of $A$ partition $\partial W$ into $2k$ polylines (each consisting of at most three axis-parallel segments), which we further call \emph{candidate curves}.

    We obtain the initial (embedded) frame tree ${\mathcal F}_0$ from the forest ${\mathcal B}_0$ by subdividing each segment $s_i$ by the two vertices $a_i$ and $b_i$ into three segments, and then adding a suitable set of candidate curves to ensure that $\mathcal{F}_0$ is connected.
    This set of candidate curves can be chosen, for example, greedily by considering them one by one and inserting each if and only if its insertion does not create a cycle.
    We further inductively construct an auxiliary sequence ${\mathcal F}_1,\dots,{\mathcal F}_k$ of $k$ frame trees having a gradually decreasing number of crossings with $\partial\bar\Sigma$, until the last of them ${\mathcal F}_k$ does not cross $\partial\bar\Sigma$ at all.
    Later, we will show how to use these trees to inductively redraw the segments $s_i$ so that their final polyline representations lie in the union of the small neighborhoods of the edges of ${\mathcal F}_k$, and thus not cross $\partial\bar\Sigma$.
    For $j\in\{1,\dots,k\}$, we obtain ${\mathcal F}_j$ from ${\mathcal F}_{j-1}$ by deleting the segment $f_j:=a_jb_j$ and adding any candidate curve $\gamma_j$ connecting the two different components of ${\mathcal F}_{j-1}-f_j$.
    Such a candidate curve must exist since the candidate curves form a cycle which intersects each of the two components of ${\mathcal F}_{j-1}-f_j$ (for example, in $a_j$ and in $b_j$, respectively).
    The last frame tree, ${\mathcal F}_k$, indeed does not cross $\partial\bar\Sigma$, as it contains none of the segments $f_j=a_jb_j$.

    It remains to describe how the frame trees  ${\mathcal F}_0,\dots,{\mathcal F}_k$ are used to change the polyline representations for the segments $s_i$ in $k$ steps so that after the last step none of them crosses $\partial\bar\Sigma$.
    We remark that the other polyline representations are not changed during the process, thus none of them crosses $\partial\bar\Sigma$ either.
    During the process, we keep the invariant that after the $j$-th step ($j\in \{1,\dots,k\}$), the polyline representation of every segment $s_i$, $i\in\{1,\ldots,k\}$, lies in the union of the small neighborhoods of the edges of the frame tree ${\mathcal F}_j$.

    We now describe the $j$-th step.
    For every $i\in\{1,\ldots,k\}$, the polyline representation of $s_i$ lies in the union of small neighborhoods of the edges of ${\mathcal F}_{j-1}$.
    Let $\rho_j$ be the boundary of a sufficiently small neighborhood of the frame tree ${\mathcal F}_j$.
    Since ${\mathcal F}_j$ is a tree, $\rho_j$ is a contractible closed curve.
    We define $\mathcal S_j$ to be the set of all connected components of the intersections of the current polyline representations of the segments $s_i$ (after the first $j-1$ steps) with the exterior of $\rho_j$.
    In other words, each element of $\mathcal S_j$ is a maximal portion of some current polyline representation that lies outside $\rho_j$ (see dashed blue and cyan segments in \cref{fig:phase1}).
    These curve portions cross $\partial\bar\Sigma$ near $f_j$ (the edge of ${\mathcal F}_{j-1}\setminus {\mathcal F}_{j}$),
    and connect points (their ``raw ends'') on $\rho_j$ near $a_j$ with points (their ``raw ends'') on $\rho_j$ near $b_j$.
    If $f_j$ crosses a vertical part of $\partial\bar\Sigma$, then these raw ends appear in the same bottom-to-top order near $a_j$ and near $b_j$.
    Otherwise, they appear in the same left-to-right order near $a_j$ and near $b_j$.
    We choose arbitrarily one of the two possible directions along~$\rho_j$ from the vicinity of $a_j$ to the vicinity of $b_j$.
    We reroute the curve portions in $\mathcal{S}_j$ so that they are newly drawn as polylines next to each other along this direction in an $\varepsilon$-neighborhood of~$\rho_j$, and each of them connects the corresponding pair of raw ends (one near $a_j$ and the other near $b_j$); see \cref{fig:phase1}.
    Due to the above-mentioned consistency in the bottom-to-top or left-to-right order, and since $\rho_j$ is a simple closed curve on an orientable surface, a sufficiently small neighborhood of~$\rho_j$ is homeomorphic to an annulus.
    Hence, the rerouted curve portions can be drawn there as pairwise disjoint parallel polylines, without crossing each other or any polyline representation in the current drawing.
    Thus, we obtain $\mathcal{B}_j$ from  $\mathcal{B}_{j-1}$ by the rerouting of the segments $\mathcal{S}_j$; see \cref{fig:phase1}.
    This completes the inductive description of how to reconfigure $\mathcal{B}$ to $\mathcal{B}_k = \mathcal{B}^*$.
    We discuss the number of moves and the runtime in Appendix~\ref{app:torus}.
\end{proof}

\paragraph*{Phase 2. Use the boundary of the fundamental square to solve the problem}
After Phase~1, no edge of $\mathcal B^*$ or $\mathcal R^*$ crosses the boundary of $\bar\Sigma$.
In Phase~2 we reconfigure $\mathcal B^*$ to $\mathcal R^*$.
We first consider the case where $F$ is a tree and then generalize to forests.

\begin{restatable}{rlemma}{reconfigurationtwoplaneembeddingstree}
    \label{lem:reconfiguration-of-two-plane-embeddings-tree}
    \renewcommand{\mylink}{\statlink{lem:reconfiguration-of-two-plane-embeddings}\prooflink{Preconfigurationtwoplaneembeddingstree}}
    Let $\mathcal{B}^*$ be a partial $\bar\Sigma$-embedding of subforest $F$ of a tree $T=(V,E)$ and $\mathcal{R}^*$ be a $\bar{\Sigma}$-embeddings of $T$ such that no edge curve of $\mathcal{B}^*$ or $\mathcal{R}^*$ crosses the boundary of the fundamental square $\bar{\Sigma}$.
    Let $\chi$ be the number of crossings between $\mathcal{B}^*$
    and $\mathcal{R}^*$  and $s^* = s(\mathcal{B}^*) + s(\mathcal{R}^*) $.
    There is a reconfiguration sequence from $\mathcal{B}^*$ to the restriction $\mathcal{R}^*[F]$ of $\mathcal{R}^*$ to $F$ of length $O(\chi + s^*)$.
    Each embedded graph in the sequence has at most $O(s^*(\chi + s^*))$ segments.
    There is an algorithm to compute the reconfiguration sequence in time $O(s^*(\chi + s^*)^2)$.
\end{restatable}

\begin{proof}
    For the reconfiguration algorithm of Phase~2,
    we iterate over the edges $e$ of $T$ in leaf-to-root order after choosing an arbitrary root for~$T$
    (i.e., each edge is processed only after all edges in the subtree below have been processed, for example, by processing the tree in a post-order traversal).
    In each stage, we \defn{fix} $e$, meaning that we reconfigure the edge curve $b = \mathcal B^* (e)$ to match $r = {\mathcal R}^*(e)$,
    or---if $b$ does not exist (which may happen if $\mathcal{B}^*$ is a partial embedding which we make use of in \cref{lem:reconfiguration-of-two-plane-embeddings})---just ensure that no edge curve of~$\mathcal B^*$ crosses $r$.
    After each stage, we ensure that no edge curve of $\mathcal B^*$ crosses the top/bottom boundary of $\bar\Sigma$, but allow right/left crossings.
    Embedding ${\mathcal R}^*$ will not change.

    The general subproblem is as follows.
    We have an edge $e = xy$ directed toward the root of $T$, with corresponding edge curves $r= \mathcal{R}^*(e)$ and $b= \mathcal{B}^*(e)$ (if it exists).
    Every edge $f$ in the subtree rooted at $x$ is a \defn{fixed} edge (that is,
    no edge of $\mathcal B^*$ crosses ${\mathcal R}^*(f)$ and, if $\mathcal B^* (f)$ exists, then
    $\mathcal B^* (f)= {\mathcal R}^*(f)$).
    No edge curve  of ${\mathcal R}^*$ crosses $\partial\bar\Sigma$.
    No edge curve of $\mathcal B^*$ crosses the top/bottom of $\bar\Sigma$.
    The goal is to fix edge $e$
    while reconfiguring only non-fixed edges of $\mathcal B^*$.

    In order to fix $e$,
    we must eliminate the crossings of $r$ with edge curves of~$\mathcal B^*$.
    The basic idea is to reroute the edge curves at crossings to detour around
    ${\mathcal R}^*_x$, the ${\mathcal R}^*$ embedding of the
    subtree rooted at~$x$.
    This is not immediately possible if $b$ exists, since such a detour would cross $b$ where it is incident to $x$.
    For example, in the first pane of Figure~\ref{fig:phase2},
    ${\mathcal B}^*(yz)$ crosses $r$, but rerouting it around
    ${\mathcal R}^*_x$, as shown in the last pane of the figure,
    would introduce a crossing with $b$.
    Furthermore, we know that we must make use of the torus.
    We therefore begin by rerouting $b$ (if it exists) around the torus.
    There are three steps (see Figure~\ref{fig:phase2}).

    \begin{enumerate}
        \item Reconfigure $b$ to a \defn{desire path}, an edge curve $d$ from $x$ to $y$  that initially goes alongside~$r$, never crosses $r$, and crosses the boundary of $\bar\Sigma$ once along the top/bottom boundary.
        \item
              Eliminate any crossing between $r$ and an edge curve of
              $\mathcal B^*$ by
              reconfiguring the  $\mathcal B^*$ piece of the crossing to walk around ${\mathcal R}^*_x$.
        \item Reconfigure $d$ to $r$. (This step is trivial and not further discussed.)
    \end{enumerate}

    \begin{figure}[t]
        \centering
        \includegraphics[width=\textwidth]{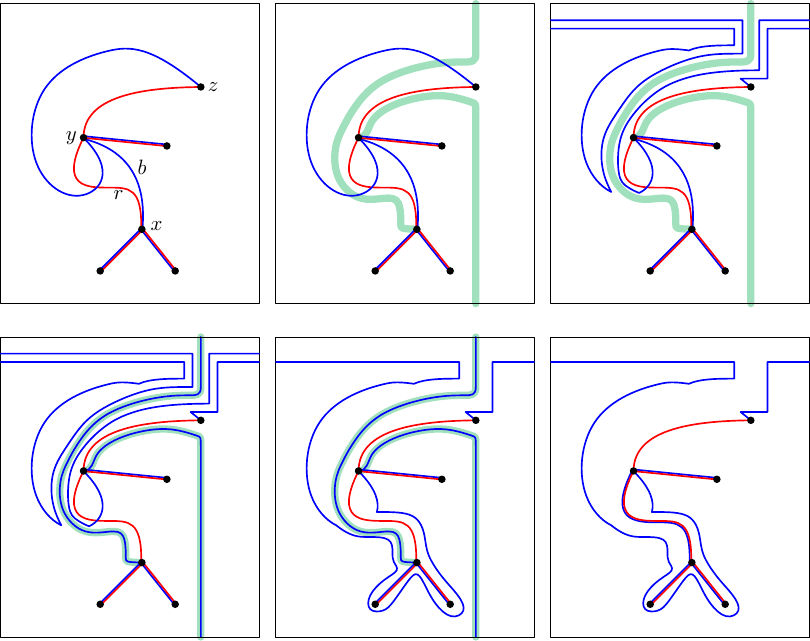}
        \caption{
        Phase 2 applied to edge $r = {\mathcal R}^* (xy)$  after fixing the subtree rooted at $x$.
        Top:
        input; Step 1a, choosing desire path $d$ (green); Step 1b, eliminating the two crossings along $d_x$.
        Bottom: Step 1c, rerouting $b$ to $d$; Step 2, eliminating the crossing along $r$; Step 3, reconfiguring $d$ to $r$.
        }
        \label{fig:phase2}
    \end{figure}

    \subparagraph{Step 1: Reconfigure $\bm{b}$ to a desire path $\bm{d}$.}

     $d$ must not cross any blue edge curve,
    but in planning $d$, we initially allow it to cross non-fixed blue edge curves and then reroute those to eliminate the crossings.
    Although not strictly necessary,
    we impose the stronger condition that the desire path $d$ does not cross any edge curve of ${\mathcal R}^*$.
    We use three substeps.

    \begin{description}
        \item[{\bf 1a.}] Construct $d$, allowing it to cross edge curves of $\mathcal B^*$ but not ${\mathcal R}^*$.

        \item[{\bf 1b.}] Eliminate any crossing between $d$ and an edge curve of
              $\mathcal B^*$ by reconfiguring the  $\mathcal B^*$ piece of the crossing to a curve that crosses the left/right boundary of $\bar\Sigma$.

        \item[{\bf 1c.}] Reconfigure $b$ to $d$.
    \end{description}

    \begin{figure}[tbh]
        \centering
        \includegraphics[scale=1.15]{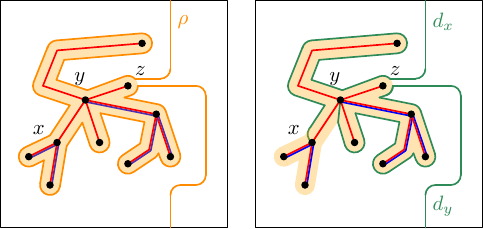}
        \caption{Step 1a: constructing $d$ for tree ${\mathcal R}^*$ rooted at $z$.
            (left) The neighborhood of ${\mathcal R}^*$ (shaded yellow) and the curve $\rho$ (yellow).
            (right) For edge $xy$, the two parts $d_x$ and $d_y$ of the desire path $d$.
        }
        \label{fig:rho}
    \end{figure}
    \subparagraph{Step 1a: Construct $\bm{d}$.}

    Let $\rho_0$ be the boundary of an $\varepsilon$-neighborhood of ${\mathcal R}^*$, where $\varepsilon>0$ is chosen small enough that $\rho_0$ is simple (i.e., non-crossing).
    Near the root vertex $z$ of the tree ${\mathcal R}^*$, remove a tiny arc of $\rho_0$ between two points
    $z_1$ and $z_2$ (in clockwise order around $\rho_0$) located just before and just after the closest approach of $\rho_0$ to $z$; by construction of $\rho_0$ as the boundary of a neighborhood, this arc does not cross ${\mathcal R}^*$.
    Augment the resulting path $\rho_0 \setminus z_1z_2$ to a closed curve $\rho$ by
    finding a path starting at $z_1$, traveling in the interior of $\bar\Sigma \setminus {\mathcal R}^*$ to a point on the top boundary of $\bar\Sigma$, crossing to the bottom boundary, and then traveling in the interior of $\bar\Sigma \setminus {\mathcal R}^*$ to $z_2$.
    See Figure~\ref{fig:rho} (left).
    Note that $\rho$ is a non-separating curve (w.r.t. $\Sigma$) and has
    $O(s({\mathcal R}^*))$
    segments -- we will use this in our analysis at the end of the proof.
    In fact, we can choose $\rho$ so that
    all but $O(1)$ of its segments
    travel alongside edge curves of ${\mathcal R}^*$, with each segment of ${\mathcal R}^*$ traversed $O(1)$ times.

    Desire path $d$ follows $\rho$ in the direction that keeps ${\mathcal R}^*$ on the right (from $z_2$ to $z_1$ along $\rho_0$, then across the top/bottom of $\bar\Sigma$), from $x$ to $y$.
    Since $\rho$ visits a neighborhood of each vertex $v$  $\mathrm{deg}(v)$ times, we must specify this more precisely.
    If we follow $\rho_0$ in this direction from $z_2$ to $z_1$, the first visit to any vertex follows the first (downward) traversal of its parent edge, and the last visit to any vertex precedes the second (upward) traversal of its parent edge.
    To construct $d$, we join $x$ to $\rho_0$ at the last visit of the edge $\{x, y\}$ (counterclockwise of $r$), then follow $\rho$ past $z_1$, across the top/bottom of $\bar\Sigma$, and past $z_2$, joining to $y$ at the last visit of the edge $\{x,y\}$;
    see Figure~\ref{fig:rho} (right).
    Observe that $d$ starts along $r$, does not cross ${\mathcal R}^*$, and crosses the boundary of $\bar\Sigma$ once on the top/bottom; i.e., $d$ is valid.

    \subparagraph{Step 1b. Eliminate crossings between $\bm{d}$ and $\bm{\mathcal{B}^*}$.}
    First, we resolve crossings on the directed portion $d_x$ from $x$ to the top of $\bar\Sigma$ in order, starting with the last one.
    Suppose $d_x$ crosses an edge curve ${\mathcal B}^* (f)$ at point $p$; see Figure~\ref{fig:step1band2}.
    Add new vertices $p_L$ and $p_R$ to ${\mathcal B}^* (f)$
    just~to the left and right of $d_x$.  Replace the part of ${\mathcal B}^* (f)$ from $p_L$ to $p_R$ by a curve traveling forward left of $d_x$ to just before the top of $\bar\Sigma$, then going left to the left boundary of $\bar\Sigma$, crossing to the other side, and continuing to just before $d_x$, then traveling backward left of $d_x$  until reaching $p_R$. This curve does not cross $d$ or any curve of $\mathcal B^*$ and each successively rerouted curve will \enquote{nest inside} the previous ones, i.e., lie closer to $d$ and closer to the top boundary of $\bar\Sigma$.
    After eliminating all crossings along $d_x$,  eliminate crossings
    on the portion of $d_y$ from $y$ to the bottom of $\bar\Sigma$ by rerouting them to cross the left/right boundary of $\bar\Sigma$ near the bottom.

    \begin{figure}[t]
        \centering
        \includegraphics[scale=1, trim = 1cm 0cm 1cm 0cm]{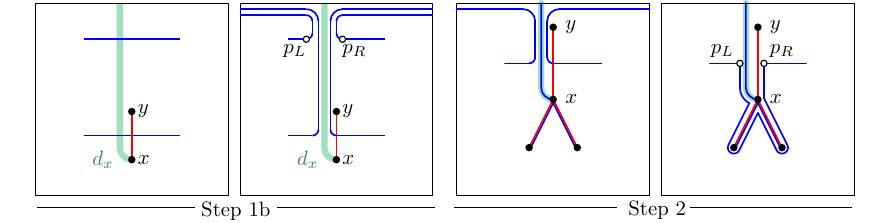}
        \caption{Steps 1b and 2: Eliminating crossings between  $\mathcal B^*$ and $d_x$ and $r$, respectively; shown schematically.
            Step 1: (left) $d_x$ is crossed by two segments of $\mathcal B^*$.
            (right) The crossings are rerouted to go around the left/right boundary of $\bar\Sigma$ near the top.
            Step 2:
            (left) $r$ is crossed by a segment of $\mathcal B^*$.
            (right) The crossed segment is rerouted to go around the fixed subtree rooted at $x$.}
        \label{fig:step1band2}
    \end{figure}

    \subparagraph{Step 1c. Reconfigure $\bm{b}$ to $\bm{d}$.}
    As $d$ is not crossed anymore, we can now reconfigure $b$ to $d$.

    \subparagraph{Step 2. Eliminate any crossing between $\bm{r}$ and an edge curve of $\bm{\mathcal{B}^*}$.}

    Direct $r$ from $x$ to $y$ and eliminate crossings in order, starting with the first one.
    Suppose that $r$ crosses an edge curve ${\mathcal B}^* (f)$ at point $p$.
    See Figure~\ref{fig:step1band2}.
    If $b$ does not exist, then let $p_L$ and $p_R$ be points on ${\mathcal B}^* (f)$ just to the left and right of $p$.
    Otherwise (if $b$ exists),
    we chose $d$ in Step 1a to start at $x$ and hug the edge curve $r$,
    so we know that ${\mathcal B}^* (f)$ used to cross $d$, and was rerouted in Step 1b to eliminate that crossing.
    Let $p_L$ be the point on ${\mathcal B}^* (f)$ where we diverted it to eliminate the crossing with $d$.
    Let $p_R$ be a point on ${\mathcal B}^* (f)$ just to the right of $r$.
    Having specified $p_L$ and $p_R$, we replace the part of ${\mathcal B}^* (f)$ from $p_L$ to $p_R$
    (which may detour along $d$ and loop around the left/right boundary of $\bar\Sigma$ near the top)
    by a curve that goes counterclockwise around ${\mathcal R}^*_x$, the fixed subtree rooted at $x$.
    This curve does not cross any curve of the current ${\mathcal B}^*$ and each successive rerouted curve will nest inside the previous ones, i.e., lie closer to ${\mathcal R}^*_x$.

\medskip
    This completes the description of the reconfiguration algorithm. We describe how to  reuse the same $\rho$ in each step and discuss the resulting complexity  in  Appendix~\ref{app:torus}.
\end{proof}

\begin{restatable}{rlemma}{forestlemma}
    \label{lem:reconfiguration-of-two-plane-embeddings}
    \renewcommand{\mylink}{\statlink{lem:reconfiguration-of-two-plane-embeddings}\prooflink{Pforestlemma}}
    Let $\mathcal{B}^*$ and $\mathcal{R}^*$ be two $\bar{\Sigma}$-embeddings of a forest $F=(V,E)$ such that no edge curve of $\mathcal{B}^*$ or $\mathcal{R}^*$ crosses the boundary of the fundamental square $\bar{\Sigma}$.
    Let $c$ be the number of connected components of $F$, let $\chi$ be the number of crossings between
    $\mathcal{B}^*$ and
    $\mathcal{R}^*$, and let $s^* = s(\mathcal{B}^*) + s(\mathcal{R}^*) $.
    Then there is a reconfiguration sequence from $\mathcal{B}^*$ to
    $\mathcal{R}^*$
    of length $O(c(\chi + s^*))$.
    Each embedded graph in the sequence has at most $O(c^2 s^*(\chi + s^*))$ segments.
    There is an algorithm to compute the reconfiguration sequence in time $O(c^3 s^*(\chi + s^*)^2)$.
\end{restatable}

\begin{proof}[Proof Sketch]
    We recursively augment $F$ to a tree $T$ and simultaneously augment $\mathcal{R}^*$  to an embedding $\mathcal{R}^A$ of $T$.
    In each step, we reduce the number of connected components by adding an augmenting edge $e_a$ that closely follows previously existing edge curves and one of the noncrossing triangulating edges chosen to connect the components of the planar straight-line graph induced by $\mathcal{R}^*$.
    Then, we apply \cref{lem:reconfiguration-of-two-plane-embeddings-tree}.
    In the process, we ignore the reconfiguration of augmenting edges, which have no corresponding curve in $\mathcal{B}^*$; i.e., $\mathcal{B}^*$ is a partial $\bar\Sigma$-embedding of $T$.
    See Appendix~\ref{app:torus} for details.
\end{proof}

\begin{proof}[Proof of \cref{thm:tree}]
    \label{Pthmforestontorus}
    From Phase 1 (Lemma~\ref{lem:making-p-plane}) $\mathcal B^*$ has $O\left(3^{k({\mathcal B})} s({\mathcal B})\right)$ segments and similarly for $\mathcal R^*$.
    In the worst case, every segment of $\mathcal B^*$ crosses every segment of $\mathcal R^*$, so $\chi$ is in $O(3^k s^2)$, where $k = k({\mathcal B}) + k({\mathcal R})$ and $s = s({\mathcal B}) + s({\mathcal R})$.
    Also $s^*$ is in $O(3^{k} s)$.
    In Phase 2 (Lemma~\ref{lem:reconfiguration-of-two-plane-embeddings}) the number of  steps is $O(c(\chi + s^*))$ which is $ O(c 3^k s^2)$.
    Each embedded graph in the  sequence has $O(c^2 s^*(\chi + s^*)) = O(c^2 3^{2k} s^3)$ segments.
    If we output an explicit list of each embedding in the sequence, the output size is $O(c^3 3^{3k}s^5)$.
    If we only output the reconfigured portion of each edge curve after each step,  the bound is $O(c^2 3^{2k} s^3)$.
    The exponential dependence on $k$ arises only in Phase 1.
    In positive terms, our algorithm is fixed parameter tractable in $k$.
\end{proof}

\noindent We discuss the following generalization to orientable surfaces of higher genus in Appendix~\ref{app:torus}.

\begin{restatable}{rcorollary}{corhighergenus}
    \label{cor:highergenus}
    \renewcommand{\cmylink}{\cstatlink{cor:highergenus}\cprooflink{Pcorhighergenus}}
    Let $F = (V,E)$ be a forest and let $\mathcal{B}$ and $\mathcal{R}$ be two $\bar{\Sigma}$-embeddings of $F$,
    where $\bar{\Sigma}$ is a fundamental polygon of genus $g\geq 1$.
    Let $c$ be the number of connected components of $F$,
    $s = s(\mathcal{B}) + s(\mathcal{R}) $ denote the total number of segments in  $\cal B$ and $\cal R$ and let $k = k(\mathcal{B}) + k(\mathcal{R})$ denote the total number of segments in $\cal B$ and $\cal R$  that intersect $\partial\bar\Sigma$.
    Then there is a reconfiguration sequence from $\mathcal{B}$ to $\mathcal{R}$ of length $O(c 3^k g^2s^2)$.
    Each embedded graph in the sequence has $O(c^2 3^{2k} g^3s^3)$ segments.
    There is an algorithm to compute the reconfiguration sequence in time proportional to the output size $O(c^3 3^{3k} g^5s^5)$.
\end{restatable}

\section{Reconfiguration of Planar Graphs on Orientable Surfaces}
\label{sec:restrictedSettings}

We study generalizations to planar graphs embedded on an orientable surface of genus $g\geq1$.

\subsection{Reconfiguration with Fixed Rotation System}
\label{ssec:embedding-preserving}

\begin{restatable}{rtheorem}{thmembeddingpreserving}
    \label{thm:embedding-preserving}
    \renewcommand{\mylink}{\statlink{thm:embedding-preserving}\prooflink{Pthmembeddingpreserving}}
    Let $\cal B$ and $\cal R$  be two $\Sigma$-embeddings of a planar graph $G$ on an orientable surface $\Sigma$ of genus $g\geq 1$. Moreover, let $\cal B$ and $\cal R$ have the same rotation system and suppose that there is a planar embedding $\cal E$ of $G$ with the same rotation system as $\cal B$ and $\cal R$. Then $\cal B$ can be reconfigured to $\cal R$. Furthermore, the rotation system of the $\Sigma$-embedding of $G$ remains the same throughout the reconfiguration sequence.
\end{restatable}

\begin{proof}[Proof Sketch]
    We give a brief overview of the proof of \Cref{thm:embedding-preserving}; the details can be found in Appendix~\ref{app:embedding-preserving}.
    Assume first that $G$ is connected. We choose a spanning tree of $G$ and reroute the edge curves of $\cal B$ and $\cal R$ to hug the spanning tree (\Cref{lem:hug-a-tree}, Appendix~\ref{app:embedding-preserving}).
    If $G$ itself is a tree, we can reduce the problem to \Cref{cor:highergenus} by first reconfiguring $\mathcal{B}$ to be identical to $\mathcal{R}$ in some small neighborhoods of the vertices (\Cref{lem:locally}, Appendix~\ref{app:embedding-preserving}), and then applying the algorithm in \Cref{sec:torus}, which does not change these neighborhoods.
    The rotation system remains unchanged (it is determined by the neighborhoods of the vertices).

    $G$ is not necessarily connected.
    We can augment its spanning forest to a spanning tree if the components each lie in small neighborhoods of their spanning trees, and they have a consistent \enquote{outer face} (\Cref{lem:hug-a-tree,lem:augmentation}, Appendix~\ref{app:embedding-preserving}).
    Finally, if $G$ is connected and $\mathcal{B}$ and $\mathcal{R}$ lie in the neighborhood of the embedding of a spanning tree $T$, we can follow the steps of the reconfiguration algorithm in \Cref{sec:torus} for $T$: Each time the algorithm reroutes an edge $e$ of $T$, we reroute a \enquote{bundle} of edges in the neighborhood of the edge curve of $e$.
\end{proof}

\subsection{Reconfiguration of Series-Parallel Graphs}
\label{ssec:series-parallel}

The family of \emph{series-parallel graphs}~\cite{DUFFIN1965303}  can be defined recursively as follows.
\begin{enumerate}
    \item The graph consisting of a single edge $st$ is a series-parallel graph with poles $s$ and $t$.
    \item Given two series-parallel graphs $G_1$ with poles $s_1$ and $t_1$ and $G_2$ with poles $s_2$ and $t_2$, the following are series-parallel: (i)~The \emph{series composition} obtained by identifying $t_1$ and $s_2$ (with poles $s_1$ and $t_2$). (ii)~The \emph{parallel composition} obtained by identifying $s_1$ and $s_2$ as well as identifying $t_1$ and $t_2$ (with poles $s_1=s_2$ and $t_1=t_2$).
\end{enumerate}

 \noindent Series-parallel graphs were studied as early as 1892~\cite{MACMAHON1994225} as they naturally occur in electrical networks~\cite{DUFFIN1965303,MACMAHON1994225,https://doi.org/10.1002/sapm194221183}.
Recently, series-parallel graphs have received attention both in graph algorithms~\cite{DBLP:journals/combinatorics/BencsHR23,DBLP:journals/acta/BlochHansenS25,DBLP:journals/cys/MedinaGMH22,DBLP:journals/dm/PanZ22} and in graph drawing~\cite{DBLP:journals/tcs/AngeliniBKM22,DBLP:conf/icaa/AzgorR25,DBLP:journals/algorithmica/DidimoKLO23,DBLP:journals/jgaa/Eppstein21}.
The families of edge-maximal series-parallel graphs and $2$-trees coincide~\cite{1992177} and the family of series-parallel graphs includes all $2$-connected graphs without $3$-connected minors~\cite{DBLP:journals/algorithmica/BienstockM90}.
Results on series-parallel graphs often cover all graphs of treewidth two and can be intermediate steps towards  results on all $2$-connected graphs --  embeddings of  $2$-connected graphs  can be enumerated using SPQR-trees that describe their construction using three operations (parallel, series and rigid)~\cite{DBLP:journals/siamcomp/BattistaT96}.

Notably, all plane embeddings of series-parallel graphs differ only in the order in which parallel subgraphs are sorted at their common poles; see e.g.~\cite{DBLP:journals/algorithmica/BienstockM90}.
We will assume without loss of generality that the final step in the recursive construction is a parallel composition (otherwise, we add a parallel composition with an edge connecting both poles at the end).

\begin{restatable}{rtheorem}{thmseriesparallel}
    \label{thm:seriesparallel}
    \renewcommand{\mylink}{\statlink{thm:seriesparallel}\prooflink{Pseriesparallel}}
    Let $G$ be a series-parallel graph and $\mathcal{B}$ and $\mathcal{R}$  two interior $\bar{\Sigma}$-embeddings of $G$ on a fundamental polygon $\bar{\Sigma}$ with $g \geq 1$. Then $\mathcal{B}$ can be reconfigured to $\mathcal{R}$.
\end{restatable}

\begin{figure}[t]
    \centering
    \includegraphics[width=\textwidth]{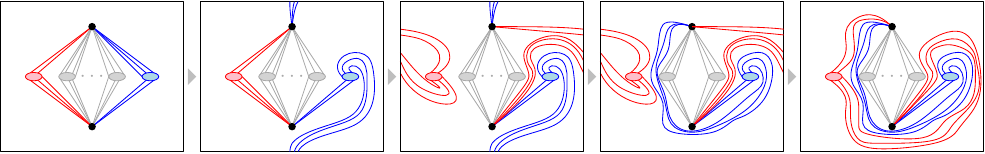}
    \caption{Reordering of two parallel components (red \& blue) based on an interior $\bar{\Sigma}$-embedding.
    }
    \label{fig:seriesparallel}
\end{figure}
\begin{proof}[Proof Sketch] For a pair of poles $s$ and $t$, we describe a procedure to exchange  two  parallel components $C_1$ and $C_2$ occurring consecutively in clockwise orientation around $s$ (and hence in counterclockwise orientation around $t$).
    We first make the face between $C_1$ and $C_2$ the outer face of the embedding using \cref{thm:embedding-preserving}, yielding the starting configuration shown in \cref{fig:seriesparallel}.
    Then, we first reconfigure the edges of $C_1$ incident to $s$ so that they intersect the vertical part of $\partial \bar{\Sigma}$, then we let the edges of $C_2$ incident to $t$ intersect the horizontal part of $\partial \bar{\Sigma}$ so that $C_1$ and $C_2$ have successfully exchanged their positions; see the second and third subfigures in \cref{fig:seriesparallel}.
    Then, we transform back to an embedding contained within the torus boundary while maintaining the rotation system.
    See Appendix~\ref{app:series-parallel} for more details.
\end{proof}

\section{Reconfiguration in the Projective Plane}
\label{sec:projective}

\begin{restatable}{rtheorem}{thmprojective}
    \label{thm:projective}
    Let $G$ be a perfect matching on $2n$ vertices, and let $\mathcal{B}$ and $\mathcal{R}$ be two $\mathbb{R}P^2$-embeddings of $G$ such that there is a single region homeomorphic to a disk that contains both embeddings.
    Then, $\mathcal{B}$ and $\mathcal{R}$ are reconfigurable.
\end{restatable}

\begin{proof}
    By applying a homeomorphism, we may assume that $\mathcal{B}$ and $\mathcal{R}$ are given in the above described representation of $\mathbb{R}P^2$, with a sphere and a crosscap $\partial D$, in which the antipodal points of the circle $\partial D$ are identified.
    The homeomorphism should map the region containing both embeddings to a disk disjoint from the crosscap. Let $G=(V,E)$ be a perfect matching on $2n$ vertices $V=\{s_1,\ldots,s_n,t_1,\ldots,t_n\}$ and $n$ edges $E=\{e_i=s_it_i: 1\leq i\leq n\}$. Let $\mathcal{B}$ and $\mathcal{R}$ be two $\mathbb{R}P^2$-embeddings of $G$ such that $\mathcal{B}(e_i)=P_i$ and $\mathcal{R}(e_i)=Q_i$ for  all $i$, $1\leq i\leq n$.

    We construct a \emph{canonical embedding} $\mathcal{M}$, where every curve passes through the crosscap exactly once, and then show how to transform $\mathcal{B}$ (resp., $\mathcal{R}$) to $\mathcal{M}$.
    We construct the canonical matching $\mathcal{M}$ as follows: Choose $2n$ equally spaced points along the circle $\partial D$, denoted $(a_1,\ldots a_n,b_1,\ldots, b_n)$ in counterclockwise order. Note that $a_i$ and $b_i$ are antipodal points in $\partial D$, hence they are identified in $\mathbb{R}P^2$.
    Since both embeddings lie in a disk disjoint from the crosscap, we may first choose pairwise disjoint arcs from the points $a_1,\ldots,a_n,b_1,\ldots,b_n$ to the boundary of that disk in the annulus between the two boundaries.
    We then extend these arcs inside the disk one by one to the vertices $s_1,\ldots,s_n,t_1,\ldots,t_n$ while avoiding the previously constructed arcs.
    Thus, for $i=1,\ldots,n$, we obtain pairwise disjoint paths $a_is_i$ and $b_it_i$ that cross the curves $P_j$ and $Q_j$, $1\leq j\leq n$, only finitely many times; moreover, each path $a_is_i$ avoids all paths $a_js_j$, $j<i$, and each path $b_it_i$ avoids all paths $a_js_j$, $1\leq j\leq n$, and $b_jt_j$, $j<i$.
    The canonical matching is now $\mathcal{M}=\{ M_i= s_ia_i\cup b_it_i\colon i\in \{1\ldots , n\}\}$.

    Next, construct a reconfiguration sequence from $\mathcal{B}$ to $\mathcal{M}$. For $i=1,\ldots, n$, replace the curve $P_i$ with $M_i$ (in several moves) as follows. Suppose all curves $P_j$, $j<i$, have already been replaced by $M_j$ (curves $P_i,\ldots, P_n$ may have been modified in the process). We seek to replace $P_i$ with $M_i$ without modifying any curve $M_j$, $j<i$.
    To do so, successively modify the curves $P_i,\ldots , P_n$ and eliminate all their crossings with $M_i$. Then, replace $P_i$ with $M_i$.
    We implement this strategy in two phases: First, eliminate all crossings of  $P_i,\ldots , P_n$ with $a_is_i$ (without introducing new crossings with $b_it_i$); then eliminate all crossings with $b_it_i$.

    Consider a \emph{thickening} $N_i$ of $\partial D\cup M_i\cup \left(\bigcup_{j< i}M_j\right)$, that is, a $\delta$-neighborhood
    for a small $\delta>0$; see \cref{fig:stages}~(upper left).
    Assume $\delta>0$ is sufficiently small so that for all $\ell>i$, every connected component of $N_i\cap P_\ell$ intersects $M_i$.
    Denote these connected components (arcs) by $\gamma_1,\ldots ,\gamma_t$ in the order in which they cross $a_is_i$ (from $a_i$ to $s_i$). Let $P(\gamma_j)$ be the curve that contains the arc $\gamma_j$ for $j=1,\ldots, t$.
    We modify $P_i,\ldots ,P_n$ in four stages.
    \begin{itemize}
        \item \textcolor{lipicsLineGray}{\sffamily \bfseries Stage~1:} A simplification step (described below) that modifies $P_i,\ldots ,P_n$ such that (i) $P_i$ does not cross $M_i$, and (ii) for every $j$, $i < j \leq n$, the curve $P_j$ crosses $M_i$ at most once.
        \item \textcolor{lipicsLineGray}{\sffamily \bfseries Stage~2:} For $\ell=1,2,\ldots ,t$, replace $\gamma_\ell$ with an arc $\gamma'_\ell$ between the same endpoints that closely follows $M_i$ (without crossing it) to $\partial D$, then crosses the crosscap, and finally closely follows $M_i\cup (\bigcup_{j<i}M_j)\cup (\bigcup_{j<\ell} P(\gamma_j))$; see \cref{fig:stages}~(top row).
        \item \textcolor{lipicsLineGray}{\sffamily \bfseries Stage~3:} Replace $P_i$ with $M_i$.
        \item \textcolor{lipicsLineGray}{\sffamily \bfseries Stage~4:} For each $\ell=t,t-1,\ldots , 1$, replace $\gamma'_\ell$ with an arc $\gamma''_\ell$ that closely follows the boundary of the neighborhood $N_i$, and goes around $s_i$; see \cref{fig:stages}~(bottom row).
    \end{itemize}

    \begin{figure}[t]
        \centering
        \includegraphics[width=\textwidth]{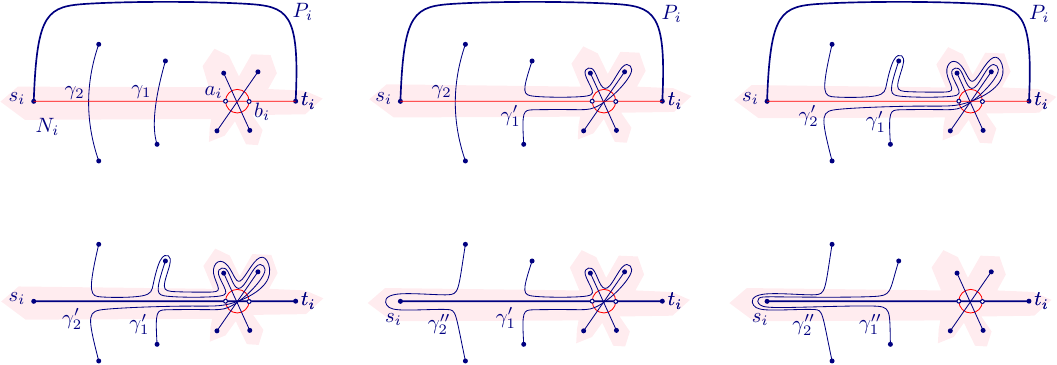}
        \caption{Replacing curve $P_i$ (blue) with $M_i$ (red). The thickening $N_i$ of $\partial D\cup M_i\cup \left(\bigcup_{j< i}M_j\right)$ is highlighted in pink. Initially, $\gamma_1$ and $\gamma_2$ cross $N_i$. Top row: Stage~2. Bottom row: Stage 4.}
        \label{fig:stages}
    \end{figure}

    It remains to describe Stage~1 (the simplification stage) in detail.
    In Stage~1, we apply the following claim while there is a curve in $\{P_i,\ldots, P_n\}$ that intersects $a_is_i$ more than once.

    \begin{restatable}{rclaim}{claimprojectiveplane}
        \label{claim:claimprojectiveplane}
        \renewcommand{\cmylink}{\cstatlink{claim:claimprojectiveplane}\cprooflink{Pclaimprojectiveplane}}

        While some curve in $\mathcal{P} = \{P_i,\ldots ,P_n\}$ intersects $M_i\setminus \{s_i\}$ at least twice,
        one can modify $\mathcal{P}$ to reduce the
        number of intersections between $M_i\setminus \{s_i\}$ and
        curves in $\mathcal{P}$.
    \end{restatable}

    We prove \cref{claim:claimprojectiveplane} in Appendix~\ref{app:projectivePlain}.
    Consequently, $\mathcal{B}$ and $\mathcal{R}$ can each be reconfigured to  $\mathcal{M}$ and combining the two reconfiguration sequences
    shows that $\mathcal{B}$ is reconfigurable to $\mathcal{R}$.
\end{proof}

\begin{remark}
    In \cref{thm:projective},  none of the edges pass through the crosscap; and we constructed an intermediate embedding in which every edge passes through the crosscap precisely once.
    It is not difficult to generalize \cref{thm:projective} to input embeddings in which one edge may pass through the crosscap arbitrarily many times, and all other edges pass through the crosscap at most once each (in this case, we can create a clean crosscap).
    However, it is unclear how to reconfigure $\mathbb{R}P^2$-embeddings where multiple edges pass through the crosscap multiple times.
\end{remark}

\section{Non-Reconfigurable Graph Embeddings on Surfaces}
\label{sec:counter}

\begin{restatable}{rtheorem}{thmnegative}
    \label{thm:negative}
    \renewcommand{\mylink}{\statlink{thm:negative}\prooflink{Pnegative}}
    Let $\Sigma$ be a surface.
    \begin{enumerate}[1.]
        \item\label{thm:negative:1} There exists a graph $H_1$ and two $\Sigma$-embeddings $\mathcal{B}_1$ and $\mathcal{R}_1$ of $H_1$ that are not reconfigurable into each other.
              This holds even if we require that $\mathcal{B}_1$ and $\mathcal{R}_1$ have the same rotation system.
        \item\label{thm:negative:2} There exist a graph $H_2$ and two rotation systems $B_2$ and $R_2$ such that $H_2$ admits $\Sigma$-embeddings implementing $B_2$ and $R_2$,
              but any pair of $\Sigma$-embeddings $\mathcal{B}_2$ and $\mathcal{R}_2$ implementing $B_2$ and $R_2$ that have the vertices of $H_2$ embedded on the same points in $\Sigma$ are not reconfigurable into each other.
    \end{enumerate}
\end{restatable}

\begin{proof}[Proof Sketch]
    For the plane (and the sphere), Ito et al.~\cite{ItoIK0MNOO25} showed that a perfect matching on four vertices $G_0$ has two embeddings, $\mathcal{B}_0$ and $\mathcal{R}_0$, that are not reconfigurable; see also \cref{fig:intro:1}. On any surface $\Sigma$, we can embed $G_0$ within a triangle $\Delta (abc)$ as shown in \cref{fig:negative} (ignoring the dashed edges), yielding the gadget graph $G_1$. For statement~\ref{thm:negative:1}, we use a uniquely embeddable triangulation $H=H(\Sigma)$ in which every $3$-cycle is facial in every embedding; its construction is given in Appendix~\ref{app:negative} using a theorem of Negami~\cite[Theorem~3.3]{Negami1985}. We identify one of its faces with $\Delta (abc)$ and insert the two embeddings of $G_1$ shown in \cref{fig:negative}, respectively.
    \begin{figure}[htbp]
        \centering
        \includegraphics[width=0.8\textwidth, page=2]{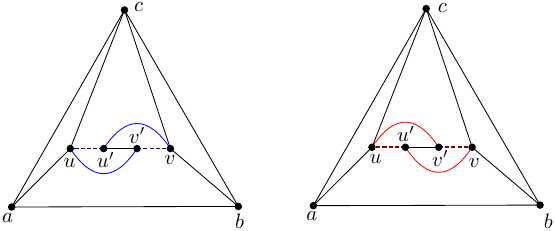}
        \caption{Illustration for the proof of \cref{thm:negative}.}
        \label{fig:negative}
    \end{figure}

    \noindent As a result, reconfiguration is not possible due to the result by Ito et al.~\cite{ItoIK0MNOO25}.
    For statement~\ref{thm:negative:2}, we augment the embeddings of $G_1$ with the two dashed edges in \cref{fig:negative}, obtaining graph $G_2$ (part of the augmented graph $H_2$).
    We then argue that in any embedding with one of the two fixed rotation systems, we must have the same topology for $G_2$.
    This again contradicts reconfigurability. See also~Appendix~\ref{app:negative}.
\end{proof}

\section{Open Problems}
\label{sec:open}
\begin{enumerate}
    \item\label{openproblem:1} We do not know if our main result for forests extends to non-orientable surfaces.
    We conjecture that a forest can always be reconfigured on any non-orientable surface.
    Phase 1 of our algorithm from \cref{sec:torus} (see \cref{lem:making-p-plane}) likely extends to non-orientable surfaces, since its shortcuts are constructed within the chosen flat representation. 
    In contrast, this extension does not seem readily available for Phase 2 (see \cref{lem:reconfiguration-of-two-plane-embeddings-tree}),
    where the rerouting argument uses that crossing an identified side of the flat representation preserves the left-to-right order of parallel arcs.
    For orientation-reversing side identifications on non-orientable surfaces, this order is reversed, so the same detour argument breaks down.
    \item\label{openproblem:2} Our result on the reconfigurability of two different interior $\bar{\Sigma}$-embeddings of series-parallel graphs may generalize to
          general planar graphs.
    The missing link is triconnected planar graphs, where reconfiguration into a mirrored embedding is required.
    \item\label{openproblem:3} Given a positive answer to Open Problem~\ref{openproblem:1}, one may wonder if we can obtain positive results for other (subclasses of) planar graphs on non-orientable surfaces.
    \item The runtime of our algorithm in \cref{sec:torus} is exponential in the complexity of the input embeddings.
    We conjecture that there is an exponential lower bound.
    We also ask what is the computational complexity of the \texttt{$\Sigma$-Embedding Reconfiguration} problem.
\end{enumerate}

\bibliographystyle{plainurl}
\bibliography{rerouting-curves}

@article{ItoIK0MNOO25,
author = {Ito, Takehiro and Iwamasa, Yuni and Kakimura, Naonori and Kobayashi, Yusuke and Maezawa, Shun-Ichi and Nozaki, Yuta and Okamoto, Yoshio and Ozeki, Kenta},
title = {Rerouting Planar Curves and Disjoint Paths},
year = {2025},
volume = {21},
number = {2},
doi = {10.1145/3715694},
journal = {ACM Trans. Algorithms},
articleno = {20},
numpages = {37}
}

@inproceedings{chambers2021morph,
  title={How to morph graphs on the torus},
  author={Chambers, Erin Wolf and Erickson, Jeff and Lin, Patrick and Parsa, Salman},
  booktitle={Proc.\ 32nd ACM-SIAM Symposium on Discrete Algorithms ({SODA})},
  pages={2759--2778},
  year={2021},
doi={10.1137/1.9781611976465.164}
}

@article{ericksonplanar,
  title={Planar and Toroidal Morphs Made Easier},
  author={Erickson, Jeff and Lin, Patrick},
  journal={Journal of Graph Algorithms and Applications},
  year={2023},
  volume={27},
  issue={2},
  pages={95--118},
  doi={10.7155/jgaa.00616}
}

@incollection{de2017computational,
  title={Computational topology of graphs on surfaces},
  author={Colin de Verdi{\`e}re, {\'E}ric},
  booktitle={Handbook of Discrete and Computational Geometry},
editor={T\'oth, Csaba D. and O'Rourke, Joseph and Goodman, Jacob E.},
  pages={605--636},
  year={2017},
chapter={23},
  publisher={Chapman and Hall/CRC},
  doi={10.1201/9781315119601}
}

@inproceedings{de2024untangling,
  title={Untangling graphs on surfaces},
  author={Colin de Verdi{\`e}re, {\'E}ric and Despr{\'e}, Vincent and Dubois, Lo{\"\i}c},
  booktitle={Proc.\ 35th ACM-SIAM Symposium on Discrete Algorithms ({SODA})},
  pages={4909--4941},
  year={2024},
doi={10.1137/1.9781611977912.175}
}

@article{DBLP:journals/dcg/ChangE17,
  author       = {Hsien{-}Chih Chang and
                  Jeff Erickson},
  title        = {Untangling Planar Curves},
  journal      = {Discret. Comput. Geom.},
  volume       = {58},
  number       = {4},
  pages        = {889--920},
  year         = {2017},
  _url          = {https://doi.org/10.1007/s00454-017-9907-6},
  doi          = {10.1007/S00454-017-9907-6},
  timestamp    = {Thu, 12 Mar 2020 17:20:54 +0100},
  biburl       = {https://dblp.org/rec/journals/dcg/ChangE17.bib},
  bibsource    = {dblp computer science bibliography, https://dblp.org}
}

@article{DBLP:journals/talg/ChangM22,
  author       = {Hsien{-}Chih Chang and
                  Arnaud de Mesmay},
  title        = {Tightening Curves on Surfaces Monotonically with Applications},
  journal      = {{ACM} Trans. Algorithms},
  volume       = {18},
  number       = {4},
  pages        = {36:1--36:32},
  year         = {2022},
  url          = {https://doi.org/10.1145/3558097},
  doi          = {10.1145/3558097},
  bibsource    = {dblp computer science bibliography, https://dblp.org}
}

@inproceedings{DBLP:conf/soda/Chang0LMSSTT18,
  author       = {Hsien{-}Chih Chang and
                  Jeff Erickson and
                  David Letscher and
                  Arnaud de Mesmay and
                  Saul Schleimer and
                  Eric Sedgwick and
                  Dylan Thurston and
                  Stephan Tillmann},
  _editor       = {Artur Czumaj},
  title        = {Tightening Curves on Surfaces via Local Moves},
  booktitle    = {Proc.\ 29th {ACM-SIAM} Symposium on Discrete
                  Algorithms ({SODA})},
  pages        = {121--135},
  _publisher    = {{SIAM}},
  year         = {2018},
  _url          = {https://doi.org/10.1137/1.9781611975031.8},
  doi          = {10.1137/1.9781611975031.8},
  timestamp    = {Tue, 02 Feb 2021 17:07:58 +0100},
  biburl       = {https://dblp.org/rec/conf/soda/Chang0LMSSTT18.bib},
  bibsource    = {dblp computer science bibliography, https://dblp.org}
}

@article{alamdari2017morph,
  title={How to morph planar graph drawings},
  author={Alamdari, Soroush and Angelini, Patrizio and Barrera-Cruz, Fidel and Chan, Timothy M. and Lozzo, Giordano Da and Battista, Giuseppe Di and Frati, Fabrizio and Haxell, Penny and Lubiw, Anna and Patrignani, Maurizio and Roselli, Vincenzo and Singla, Sahil and Wilkinson, Bryan T.},
  journal={SIAM Journal on Computing},
  volume={46},
  number={2},
  pages={824--852},
  year={2017},
doi={10.1137/16M1069171}
}

@article{DBLP:journals/algorithmica/BienstockM90,
  author       = {Daniel Bienstock and
                  Clyde L. Monma},
  title        = {On the Complexity of Embedding Planar Graphs To Minimize Certain Distance
                  Measures},
  journal      = {Algorithmica},
  volume       = {5},
  number       = {1},
  pages        = {93--109},
  year         = {1990},
  url          = {https://doi.org/10.1007/BF01840379},
  doi          = {10.1007/BF01840379},
  bibsource    = {dblp computer science bibliography, https://dblp.org}
}

@article{DBLP:journals/algorithms/Nishimura18,
  author       = {Naomi Nishimura},
  title        = {Introduction to Reconfiguration},
  journal      = {Algorithms},
  volume       = {11},
  number       = {4},
  pages        = {52},
  year         = {2018},
  _url          = {https://doi.org/10.3390/a11040052},
  doi          = {10.3390/A11040052},
  timestamp    = {Tue, 14 Aug 2018 12:19:16 +0200},
  biburl       = {https://dblp.org/rec/journals/algorithms/Nishimura18.bib},
  bibsource    = {dblp computer science bibliography, https://dblp.org}
}

@article{DBLP:journals/siamcomp/BattistaT96,
  author       = {Giuseppe Di Battista and
                  Roberto Tamassia},
  title        = {On-Line Planarity Testing},
  journal      = {{SIAM} J. Comput.},
  volume       = {25},
  number       = {5},
  pages        = {956--997},
  year         = {1996},
  url          = {https://doi.org/10.1137/S0097539794280736},
  doi          = {10.1137/S0097539794280736},
  bibsource    = {dblp computer science bibliography, https://dblp.org}
}

@incollection{DBLP:books/cu/p/Heuvel13,
  author       = {Jan {van den Heuvel}},
  editor       = {Simon R. Blackburn and
                  Stefanie Gerke and
                  Mark Wildon},
  title        = {The complexity of change},
  booktitle    = {Surveys in Combinatorics 2013},
  series       = {London Mathematical Society Lecture Note Series},
  volume       = {409},
  pages        = {127--160},
  publisher    = {Cambridge University Press},
  year         = {2013},
  url          = {https://doi.org/10.1017/CBO9781139506748.005},
  doi          = {10.1017/CBO9781139506748.005},
  bibsource    = {dblp computer science bibliography, https://dblp.org}
}

@book{viro2024elementary,
  title={Elementary Topology: Problem Textbook},
  author={Viro, Oleg Yanovich and Ivanov, Oleg Aleksandrovich and Netsvetaev, Nikita Yur\'{e}vich and Kharlamov, Viatcheslav Mikha\v{\i}lovich},
  volume={54},
  year={2024},
  publisher={American Mathematical Society},
  doi={10.1090/mbk/054}
}

@book{MT2001,
  author       = {Bojan Mohar and
                  Carsten Thomassen},
  title        = {Graphs on Surfaces},
  series       = {Johns Hopkins Series in the Mathematical Sciences},
  publisher    = {Johns Hopkins University Press},
  year         = {2001},
  doi          = {10.56021/9780801866890},
  _url          = {http://jhupbooks.press.jhu.edu/ecom/MasterServlet/GetItemDetailsHandler?iN=9780801866890\&qty=1\&source=2\&viewMode=3\&loggedIN=false\&JavaScript=y},
url={https://users.fmf.uni-lj.si/mohar/Book.html},
  isbn         = {978-0-8018-6689-0},
  timestamp    = {Tue, 19 Feb 2013 14:35:24 +0100},
  biburl       = {https://dblp.org/rec/books/daglib/0030489.bib},
  bibsource    = {dblp computer science bibliography, https://dblp.org}
}

@article{Negami1985,
  author  = {Seiya Negami},
  title   = {Construction of Graphs Which Are Not Uniquely and Not Faithfully Embeddable in Surfaces},
  journal = {Yokohama Mathematical Journal},
  volume  = {33},
  number  = {1--2},
  pages   = {67--91},
  year    = {1985},
  url     = {https://ynu.repo.nii.ac.jp/records/6751}
}

@article{MACMAHON1994225,
title = {The combinations of resistances},
journal = {Discrete Applied Mathematics},
volume = {54},
number = {2},
pages = {225--228},
year = {1994},
issn = {0166-218X},
doi = {https://doi.org/10.1016/0166-218X(94)90024-8},
_url = {https://www.sciencedirect.com/science/article/pii/0166218X94900248},
author = {Percy A. Macmahon}
}

@article{https://doi.org/10.1002/sapm194221183,
author = {Riordan, John and Shannon, Claude E.},
title = {The Number of Two-Terminal Series-Parallel Networks},
journal = {Journal of Mathematics and Physics},
volume = {21},
number = {1-4},
pages = {83--93},
doi = {https://doi.org/10.1002/sapm194221183},
_url = {https://onlinelibrary.wiley.com/doi/abs/10.1002/sapm194221183},
_eprint = {https://onlinelibrary.wiley.com/doi/pdf/10.1002/sapm194221183},
year = {1942}
}

@article{DUFFIN1965303,
title = {Topology of series-parallel networks},
journal = {Journal of Mathematical Analysis and Applications},
volume = {10},
number = {2},
pages = {303--318},
year = {1965},
issn = {0022-247X},
doi = {https://doi.org/10.1016/0022-247X(65)90125-3},
_url = {https://www.sciencedirect.com/science/article/pii/0022247X65901253},
author = {Richard J. Duffin}
}

@inproceedings{DBLP:conf/icaa/AzgorR25,
  author       = {Sk Ruhul Azgor and
                  Md. Saidur Rahman},
  _editor       = {Subhas C. Nandy and Rajat K. De and Prosenjit Gupta},
  title        = {On the Rique Number of Series-Parallel Graphs and Planar Bipartite
                  Graphs},
  booktitle    = {Proc.\ 2nd International Conference on Applied Algorithms ({ICAA})},
  series       = {LNCS},
  volume       = {15505},
  pages        = {3--14},
  publisher    = {Springer},
  year         = {2025},
  _url          = {https://doi.org/10.1007/978-3-031-84543-7_1},
  doi          = {10.1007/978-3-031-84543-7_1},
  timestamp    = {Fri, 09 May 2025 20:28:50 +0200},
  biburl       = {https://dblp.org/rec/conf/icaa/AzgorR25.bib},
  bibsource    = {dblp computer science bibliography, https://dblp.org}
}

@article{DBLP:journals/algorithmica/DidimoKLO23,
  author       = {Walter Didimo and
                  Michael Kaufmann and
                  Giuseppe Liotta and
                  Giacomo Ortali},
  title        = {Computing Bend-Minimum Orthogonal Drawings of Plane Series-Parallel
                  Graphs in Linear Time},
  journal      = {Algorithmica},
  volume       = {85},
  number       = {9},
  pages        = {2605--2666},
  year         = {2023},
  _url          = {https://doi.org/10.1007/s00453-023-01110-6},
  doi          = {10.1007/S00453-023-01110-6},
  timestamp    = {Wed, 01 Nov 2023 08:59:33 +0100},
  biburl       = {https://dblp.org/rec/journals/algorithmica/DidimoKLO23.bib},
  bibsource    = {dblp computer science bibliography, https://dblp.org}
}

@article{DBLP:journals/tcs/AngeliniBKM22,
  author       = {Patrizio Angelini and
                  Michael A. Bekos and
                  Philipp Kindermann and
                  Tamara Mchedlidze},
  title        = {On mixed linear layouts of series-parallel graphs},
  journal      = {Theor. Comput. Sci.},
  volume       = {936},
  pages        = {129--138},
  year         = {2022},
  _url          = {https://doi.org/10.1016/j.tcs.2022.09.019},
  doi          = {10.1016/J.TCS.2022.09.019},
  timestamp    = {Sun, 13 Nov 2022 17:53:45 +0100},
  biburl       = {https://dblp.org/rec/journals/tcs/AngeliniBKM22.bib},
  bibsource    = {dblp computer science bibliography, https://dblp.org}
}

@article{DBLP:journals/jgaa/Eppstein21,
  author       = {David Eppstein},
  title        = {Bipartite and Series-Parallel Graphs Without Planar Lombardi Drawings},
  journal      = {J. Graph Algorithms Appl.},
  volume       = {25},
  number       = {1},
  pages        = {549--562},
  year         = {2021},
  _url          = {https://doi.org/10.7155/jgaa.00571},
  doi          = {10.7155/JGAA.00571},
  timestamp    = {Mon, 14 Feb 2022 17:12:42 +0100},
  biburl       = {https://dblp.org/rec/journals/jgaa/Eppstein21.bib},
  bibsource    = {dblp computer science bibliography, https://dblp.org}
}

@article{DBLP:journals/acta/BlochHansenS25,
  author       = {Andrew Bloch{-}Hansen and
                  Roberto Solis{-}Oba},
  title        = {The thief orienteering problem on 2-terminal series-parallel graphs},
  journal      = {Acta Informatica},
  volume       = {62},
  number       = {2},
  pages        = {18},
  year         = {2025},
  _url          = {https://doi.org/10.1007/s00236-025-00486-y},
  doi          = {10.1007/S00236-025-00486-Y},
  timestamp    = {Mon, 07 Apr 2025 12:58:05 +0200},
  biburl       = {https://dblp.org/rec/journals/acta/BlochHansenS25.bib},
  bibsource    = {dblp computer science bibliography, https://dblp.org}
}

@article{DBLP:journals/combinatorics/BencsHR23,
  author       = {Ferenc Bencs and
                  Jeroen Huijben and
                  Guus Regts},
  title        = {On the Location of Chromatic Zeros of Series-Parallel Graphs},
  journal      = {Electron. J. Comb.},
  volume       = {30},
  number       = {3},
  year         = {2023},
  _url          = {https://doi.org/10.37236/11204},
  doi          = {10.37236/11204},
  timestamp    = {Fri, 18 Aug 2023 13:01:54 +0200},
  biburl       = {https://dblp.org/rec/journals/combinatorics/BencsHR23.bib},
  bibsource    = {dblp computer science bibliography, https://dblp.org}
}

@article{DBLP:journals/cys/MedinaGMH22,
  author       = {Marco A. L{\'{o}}pez Medina and
                  J. Leonardo Gonz{\'{a}}lez{-}Ruiz and
                  Jos{\'{e}} Raymundo Marcial{-}Romero and
                  Jos{\'{e}} Antonio Hern{\'{a}}ndez Serv{\'{\i}}n},
  title        = {Computing the Clique-Width on Series-Parallel Graphs},
  journal      = {Computaci{\'{o}}n y Sistemas},
  volume       = {26},
  number       = {2},
  year         = {2022},
  doi          = {10.13053/CYS-26-2-4250},
  _url          = {https://cys.cic.ipn.mx/ojs/index.php/CyS/article/view/4250},
  timestamp    = {Fri, 14 Mar 2025 11:51:31 +0100},
  biburl       = {https://dblp.org/rec/journals/cys/MedinaGMH22.bib},
  bibsource    = {dblp computer science bibliography, https://dblp.org}
}

@article{DBLP:journals/dm/PanZ22,
  author       = {Zhishi Pan and
                  Xuding Zhu},
  title        = {The circular chromatic numbers of signed series-parallel graphs},
  journal      = {Discret. Math.},
  volume       = {345},
  number       = {3},
  pages        = {112733},
  year         = {2022},
  _url          = {https://doi.org/10.1016/j.disc.2021.112733},
  doi          = {10.1016/J.DISC.2021.112733},
  timestamp    = {Tue, 08 Feb 2022 10:43:36 +0100},
  biburl       = {https://dblp.org/rec/journals/dm/PanZ22.bib},
  bibsource    = {dblp computer science bibliography, https://dblp.org}
}

@incollection{1992177,
title = {Chapter 5 Polynomially Solvable Cases},
author = {Frank K. Hwang and Dana S. Richards and Pawel Winter},
series = {Annals of Discrete Mathematics},
publisher = {Elsevier},
volume = {53},
pages = {177-188},
year = {1992},
booktitle = {The Steiner Tree Problem},
issn = {0167-5060},
doi = {https://doi.org/10.1016/S0167-5060(08)70202-4},
_url = {https://www.sciencedirect.com/science/article/pii/S0167506008702024}
}

@inproceedings{ChangEX15,
  author       = {Hsien{-}Chih Chang and
                  Jeff Erickson and
                  Chao Xu},
  editor       = {Piotr Indyk},
  title        = {Detecting Weakly Simple Polygons},
  booktitle    = {Proc.\ 26th {ACM-SIAM} Symposium on Discrete
                  Algorithms (SODA)},
  pages        = {1655--1670},
  _publisher    = {{SIAM}},
  year         = {2015},
  _url          = {https://doi.org/10.1137/1.9781611973730.110},
  doi          = {10.1137/1.9781611973730.110},
  timestamp    = {Tue, 02 Feb 2021 17:07:32 +0100},
  biburl       = {https://dblp.org/rec/conf/soda/ChangEX15.bib},
  bibsource    = {dblp computer science bibliography, https://dblp.org}
}

@article{AkitayaAET17,
  author       = {Hugo A. Akitaya and
                  Greg Aloupis and
                  Jeff Erickson and
                  Csaba D. T{\'{o}}th},
  title        = {Recognizing Weakly Simple Polygons},
  journal      = {Discret. Comput. Geom.},
  volume       = {58},
  number       = {4},
  pages        = {785--821},
  year         = {2017},
  _url          = {https://doi.org/10.1007/s00454-017-9918-3},
  doi          = {10.1007/S00454-017-9918-3},
  timestamp    = {Sat, 30 Sep 2023 10:11:35 +0200},
  biburl       = {https://dblp.org/rec/journals/dcg/AkitayaAET17.bib},
  bibsource    = {dblp computer science bibliography, https://dblp.org}
}

@article{hatcher1980presentation,
  title={A presentation for the mapping class group of a closed orientable surface},
  author={Hatcher, Allen and Thurston, William},
  journal={Topology},
  volume={19},
  number={3},
  pages={221--237},
  year={1980},
  publisher={Pergamon},
  doi={10.1016/0040-9383(80)90009-9}
}

@article{lackenby2024bounds,
  author = {Lackenby, Marc and Yazdi, Mehdi},
  title = {Bounds for the number of moves between pants decompositions, and between triangulations},
  journal = {Journal of Topology and Analysis},
  pages = {1--68},
  year = {2026},
  doi = {10.1142/S1793525326500056},
}

\newpage

\appendix

\section{Omitted Material from \cref{sec:torus}}
\label{app:torus}


\lemmakingpplane*

\begin{proof}[Proof of Reconfiguration Sequence Length and Runtime]%
    \label{Plemmakingpplane}

    Recall from the proof of \cref{lem:making-p-plane} in \cref{sec:torus} that $\mathcal{F}_j$ is the frame tree after $j$ induction steps.
    Let $\mathcal{B}_j$ denote the current embedding of the forest $F$ after $j$ reconfiguration steps (with $\mathcal{B}_k = \mathcal{B}^*$).
    We begin by bounding the size of ${\mathcal B}_j$ (in particular ${\mathcal B}^* = \mathcal{B}_k$, corresponding to $\mathcal{F}^*=\mathcal{F}_k$) and the length of the reconfiguration sequence.
    We first analyze $\mathcal F_j$, the number of its intersections $k(\mathcal{F}_j)$ with the boundary $\partial\bar{\Sigma}$ and the number of its segments $s(\mathcal{F}_j)$.
    Initially, $k({\mathcal F_0}) = k$ and $s({\mathcal F_0}) \in O(s)$.
    Thus, there are $k$ induction steps, indexed by $j\in\{1,\ldots,k\}$. Subdividing each segment $s_i$ at $a_i$ and $b_i$, for $i\in\{1,\ldots,k\}$, increases the number of segments by $2k$, and in each of the $k$ inductive steps, we add at most $3$ additional segments while losing the segment $f_j$.
    Hence, since $k \le s$, all $\mathcal F_j$, including ${\mathcal F}^*$, have $O(s)$ segments.

    Initially, there is at most one edge segment of $\mathcal{B} $ in the neighborhood of each segment of $\mathcal{F}_{0}$.
    Each step at most triples the number of segments of $\mathcal{B}_j$ (compared to the number for $\mathcal{B}_{j-1}$ in the neighborhood of a segment of the frame tree  $\mathcal{F}_j$ (because the curve $\rho_j$ hugs both sides of each edge of the frame tree).
    In fact, we can refine the polyline representations with this amount of segments: To see this,  consider a very small circle at each vertex of the frame tree.
    Between any two adjacent vertices $u,v$ of the frame tree, edge curves are straight segments between the corresponding two circles, going parallel to each other in a very tiny neighborhood of the straight-line segment $uv$.
    Inside each circle, the corresponding pairs of ``raw ends'' (lying on the circle) are connected by straight-line segments.
    Thus, in the $j$-th step, at most $3^j$ segments of $\mathcal{B}_j$ are reconfigured.
    Summation over $j = 1, \ldots, k$ gives $O(3^k)$ reconfiguration steps in total.
    The final ${\mathcal B}^*$ has $O(3^k)$ segments in the neighborhood of each of the $O(s)$ segments of ${\mathcal F}^*$ for a total of $O(3^k s)$ segments.

    Note that the proof provides an algorithm.
    If we output an explicit list of each embedded graph in the reconfiguration sequence, then the output size is $O(3^{2k} s)$, which dominates the runtime.
    If we instead output just the reconfigured portion of each edge curve after each reconfiguration step, then the bound is $O(3^k s)$.
\end{proof}


\reconfigurationtwoplaneembeddingstree*

\begin{proof}[Proof of Reconfiguration Sequence Length and Runtime]
    \label{Preconfigurationtwoplaneembeddingstree}

    In the remainder of this section, we describe and analyze a more efficient approach to the main algorithm that reuses the same $\rho$ in each step.
    We must specify more precisely what happens to the reconfigurations performed in Step 1b; some of these must be undone but most can remain.

    We first classify all the reconfiguration steps.
    When we eliminate a crossing between an edge curve of $\mathcal B^*$ and $d_x$ in Step 1b, we call this a \defn{top loop detour}.
    Similarly, a \defn{bottom loop detour} eliminates a crossing between an edge curve of $\mathcal B^*$ and $d_y$.
    Reconfigurations performed during Step 2 are \defn{tree detours}.
    Each tree detour eliminates one crossing between the initial $\mathcal{B}^*$ and ${\mathcal R}^*$ so the total number is at most $\chi$.
    The only other reconfigurations performed by the algorithm are the ones from Step 1c ($b$ to $d$) and Step 3 ($d$ to $r$), so their total number is at most $2|E|$ which is in $O(s^*)$.
    We describe how to reduce the number of top and bottom loop detours by reusing $\rho$ and refining the order in which edges are fixed.
    Order the edges of the tree according to the last visit along the boundary $\rho_0$ of the $\varepsilon$-neighborhood of ${\mathcal R}^*$ going clockwise starting at the root.  Let $xy$ and $x' y'$ be two successive edges in this ordering and let $d$ and $d'$ be the corresponding desire paths constructed from~$\rho$.
    Let $\rho_x$ be the point where desire path $d$ joins $x$ to $\rho$, and let $\rho_y$ be the point where desire path $d$ joins $y$ to $\rho$.
    Define $\rho_{x'}$ and $\rho_{y'}$ similarly.
    Either $x' y'$ is the parent edge of $y$, or $x'y'$ lies in the subtree rooted at $y$ (possibly with $y'=y$).
    In the first case $\rho_{y'}$ comes before $\rho_y$ along $\rho_0$, and in the second case $\rho_{y'}$ comes after (or is equal to) $\rho_y$.
    In either case $\rho_{x'}$ comes after $\rho_x$.

    We first consider top loop detours.
    Let $x^0 y^0$ be the first tree edge in the ordering and let~$d^{\,0}$ be the corresponding desire path.
    We perform all the top loop detours for $d^{\,0}$.
    Consider what happens when the algorithm progresses from edge $xy$ to the next edge $x' y'$.
    Because  $\rho_{x'}$ comes after $\rho_x$, all the top loop detours for $d'$ have already been performed.
    However, the top loop detours corresponding to crossings of $\rho$ between $\rho_{x'}$ and $\rho_{x}$ travel through the gap between $x'$ and $\rho_{x'}$ so we must revert those top loop detours in order to join $x'$ to $\rho_{x'}$.
    We perform the reversions in order from $\rho_{x}$ to $\rho_{x'}$.

    The total number of reconfiguration steps to perform top loop detours and reversions is~$O(t)$, where $t$ is
    the number of top loop detours performed for the first edge $x^0 y^0$.
    Thus~$t$ is bounded by the number of crossings between the original $\mathcal B^*$ and $\rho$.
    As noted in Step 1a, we can choose $\rho$ so that apart from $O(1)$ \enquote{free space} segments all of its segments
    travel alongside edge curves of ${\mathcal R}^*$, with each segment of ${\mathcal R}^*$ traversed $O(1)$ times.
    There are $O(s^*)$ crossings between $\mathcal B^*$ and the free space segments of $\rho$.
    Any crossing between $\mathcal B^*$ and the other segments of $\rho$ is due either to a crossing between $\mathcal B^*$ and ${\mathcal R}^*$ or to
    an incidence between an edge curve of
    $\mathcal B^*$ and a vertex.
    Thus $t$ is in $O(\chi + s^*)$.

    We now consider bottom loop detours.
    When $y'$ is the parent of $y$, there are no new bottom loop detours and we revert the bottom loop detours corresponding to crossings of $\rho$ along $r$, which are then eliminated once and for all by the tree detours in Step 2.
    Otherwise (if $y'$ is not the parent of $y$),
    we perform new bottom loop detours for crossings of $\rho$ between $\rho_y$ and $\rho_{y'}$.
    The total number of reconfiguration steps to perform bottom loop detours and reversions is bounded by the number of crossings
    between the original
    $\mathcal B^*$ and $\rho$, which, as argued above, is in
    $O(\chi + s^*)$.
    Taking into account all the reconfiguration steps, the length of the reconfiguration sequence from $\mathcal B^*$ to $\mathcal R^*$ is then $O(\chi + s^*)$.
    Finally, as noted above, $\rho$ is a polyline consisting of
    $O(s({\mathcal R}^*))\subseteq O(s^*)$ straight-line segments, and every detour uses a portion of $\rho$, so the same bound holds for each reconfigured edge curve.
    Thus each embedded graph in the reconfiguration sequence has $O(s^*(\chi+s^*))$
    segments.
    If we output an explicit list of each embedded graph in the reconfiguration sequence, then the output size is
    $O(s^*(\chi + s^*)^2)$.
    This dominates the runtime of the Phase 2 algorithm.
\end{proof}

\forestlemma*
\label{Pforestlemma}

\begin{proof}
    We will reduce to the case where $F$ is a tree so as to be able to apply \cref{lem:reconfiguration-of-two-plane-embeddings-tree}.
    Hence, as a preliminary step, we add edges to the forest $F$ to create a tree $T$,
    and augment $\mathcal R^*$ to an embedding ${\mathcal R}^A$ of tree~$T$.
    After that we can apply \cref{lem:reconfiguration-of-two-plane-embeddings-tree}.

    Subdivide the edge curves of $\mathcal R^*$ at their bends and regard their union, together with the isolated vertices of $F$, as a graph $P$ embedded with a plane straight-line embedding in the interior of $\bar\Sigma$.
    Extend $P$ to a triangulation by adding edges.
    Since the triangulation is connected, it contains a set $A$ of $c-1$ added edges that connect the $c$ components of $P$ into a tree.
    Process $A$ in any order.
    For $a=pq\in A$, let $u=p$ if $p$ is a vertex of $F$; otherwise, $p$ lies on an original edge curve $\mathcal R^*(e)$, and we choose an endpoint $u$ of $e$.
    Define $v$ analogously from $q$.
    Add the edge $e_a=uv$ to the current forest and draw it by following closely along the relevant original edge curve from $u$ to $p$ when necessary, then following $a$, and proceeding analogously from $q$ to $v$ along the corresponding original edge curve.
    The routes alongside existing edges can be chosen consistently so that the augmenting curves are pairwise disjoint.
    Since each edge of $A$ joins two current components, after processing $A$ we obtain a tree $T$ and an augmented embedding $\mathcal R^A$.
    Note that we do not explicitly construct an augmentation of $\mathcal{B}^*$ as the edges augmenting $F$ to $T$ are not actually part of the input (that is, $\mathcal{B}^*$ is a partial $\bar\Sigma$-embedding of $T$). That is, we can ignore them when we must reroute edges of $\mathcal{B}^*$ when applying \cref{lem:reconfiguration-of-two-plane-embeddings-tree}.
    On the other hand, we can make use of the $\mathcal{R}^A$-embedding of $T$ to construct the desire paths as discussed in the proof of \cref{lem:reconfiguration-of-two-plane-embeddings-tree}.
    It remains to investigate how the augmentation affects the runtime of \cref{lem:reconfiguration-of-two-plane-embeddings-tree}.

    We added $c-1$ edges to $F$. Each augmenting edge curve follows portions of at most two edge curves of $\mathcal R^*$ and one triangulating edge in $A$, hence it has $O(s({\mathcal R}^*))$ segments. Thus $s({\mathcal R}^A)$ is in $O(c s^*)$.
    For each augmenting edge curve,
    all but one of its segments follow segments of
    ${\mathcal R}^*$;
    the remaining segment runs alongside a triangulating edge in $A$.
    Thus the number of crossings between $\mathcal B^*$ and ${\mathcal R}^A$,
    which we denote by $\chi^A$, is in $O(c ( \chi + s^*))$.
    Applying \cref{lem:reconfiguration-of-two-plane-embeddings-tree} to $\mathcal B^*$ and $\mathcal R^A$ gives a sequence of length $O(\chi^A+s(\mathcal B^*)+s(\mathcal R^A))$, which is
    $O(c ( \chi + s^*))$.
    Finally, as noted above, $\rho$ is a polyline consisting of
    $O(s({\mathcal R}^A)) \subseteq O(c s^*)$ straight-line segments, and every detour uses a portion of $\rho$, so the same bound holds for each reconfigured edge curve.
    Thus, each embedded graph in the reconfiguration sequence has $O(c^2 s^*(\chi + s^*))$ segments.
    If we output an explicit list of each embedded graph in the reconfiguration sequence, then the output size is $O(c^3 s^*(\chi + s^*)^2)$.
\end{proof}


\corhighergenus*

\begin{proof}
    \label{Pcorhighergenus}
    It is easy to verify that the correctness of all techniques in this section does not depend on the fact that the fundamental polygon is of genus $1$.
    In Phase 1 (cf. \cref{lem:making-p-plane}), we define shortcuts $\gamma$ only within the region bounded by $\partial\bar\Sigma$.
    This can be done the same way if $\bar\Sigma$ is of genus $g > 1$.
    On the other hand, then the detour path along the boundary of the fundamental polygon might consist of up to $2g+3$ segments.
    Thus, the final frame tree $\cal F^*$ has at most $(2g+3)k=O(g\cdot s)$ segments.
    Thus, $\cal B^*$ and $\cal R^*$ have $O(3^k\cdot g\cdot s)$ segments and a reconfiguration sequence can be computed in $O(3^{2k} \cdot g\cdot s)$ time.
    In Phase 2, we first provided an algorithm to reroute two embeddings of a tree (cf. \cref{lem:reconfiguration-of-two-plane-embeddings-tree}). Here, we can choose two arbitrary sides $S_1$ and $S_2$ of the boundary of $\bar \Sigma$ such that we let the desire curves $d$ intersect with side $S_1$ and the rerouted blue segments intersect with side $S_2$.
    Finally, we also provided an augmentation from the forest to the tree case (cf. \cref{lem:reconfiguration-of-two-plane-embeddings}).
    Here, we augment along noncrossing triangulating edges that connect the components of the planar straight-line graph induced by $\mathcal R^*$ in the interior of $\bar\Sigma$.
    Since this construction takes place entirely in the interior of the fundamental polygon, this works independently of the genus $g \geq 1$.
    Thus, the analysis of Phase 2 remains unchanged.
    However, because of the bound from Phase 1, we now have $s^* = O(3^k\cdot g\cdot s)$ and $\chi = O(3^k\cdot g^2s^2)$.
    Thus, Phase 2 creates a reconfiguration sequence of length $O(3^k\cdot g^2\cdot s^2)$, each embedded graph in the sequence has at most $O(3^k \cdot g^3\cdot s^3)$ segments, and there is an algorithm computing the sequence in $O(3^k\cdot g^5\cdot s^5)$ time.
    Thus, all techniques also work for any fundamental polygon at the cost of a runtime that depends on $g$.
\end{proof}


\section{Omitted Material from \cref{ssec:embedding-preserving}}
\label{app:embedding-preserving}

\thmembeddingpreserving*

\label{Pthmembeddingpreserving}

\subparagraph{Preliminaries.}
We need some preparation. We would like to extend the reconfiguration sequence of a spanning tree (\Cref{sec:torus}) to \enquote{bundles} of edges. In general, we can treat the embedding of a spanning tree of $G$ similarly to the frame tree in \Cref{sec:torus}. However, for a clear comparison between the embeddings $\mathcal{B}(T)$ and $\mathcal{B}(G)$, we use the machinery developed by Chang, Erickson, and Xu~\cite{ChangEX15} (see also~\cite{AkitayaAET17}); see \Cref{fig:stripsystem}.

\begin{figure}[htbp]
    \centering
    \includegraphics[width=\textwidth]{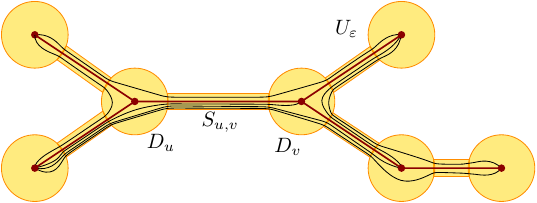}
    \caption{An $\varepsilon$-strip system of a tree embedded in the plane.}
    \label{fig:stripsystem}
\end{figure}

Let $T=(V,E(T))$ be a spanning tree (\emph{frame tree}), $\mathcal{F}$ a $\Sigma$-embedding of $T$, and let $\varepsilon>0$ be a sufficiently small real number.
An \emph{$\varepsilon$-strip system} for $\mathcal{F}$ is a decomposition of
a neighborhood of $\mathcal{F}$ into the following disks and strips.
\begin{itemize}
    \item For every vertex $v\in V$, let $D_v\subset \Sigma$ denote the disk of radius $\varepsilon$ centered at the point $\mathcal{B}(v)\in \Sigma$.
    \item For every edge $uv\in E(T)$, let $S_{uv}$ denote the strip of points with distance at most $\varepsilon^2$ from $\mathcal{B}(uv)$ that do not lie in the interior of $D_u$ or $D_v$.
\end{itemize}
The circular arcs $A_{u,v} = S_{uv}\cap D_u$ and $A_{v,u} = S_{uv}\cap D_{v}$ are called the ends of $S_{uv}$. We assume $\varepsilon>0$ is sufficiently
small that these disks and strips are pairwise disjoint except that each strip intersects exactly two disks at its ends. Finally, let $U_\varepsilon \subset \Sigma$ denote the union of all these disks and strips, which is homeomorphic to a disk since $\mathcal{F}$ is contractible.

An embedding $\mathcal{B}$ of $G=(V,E)$ is \emph{in the $\varepsilon$-strip system of $\mathcal{F}$} if, for every edge $e\in E$,
\begin{itemize}
    \item $\mathcal{B}(e)\subseteq U_\varepsilon$, and
    \item every connected component of $\mathcal{B}(e)\cap S_{uv}$ is a simple curve between the ends $A_{u,v}$ and $A_{v,u}$.
\end{itemize}

\subparagraph{Detailed Description.}
We first show that a connected planar graph can be reconfigured to a neighborhood of a spanning tree (which is a topological disk $D$), and we can even specify the outer face (w.r.t.\ the disk $D$).

\begin{lemma}\label{lem:hug-a-tree}
    Let $\mathcal{B}$ be a $\Sigma$-embedding of a connected planar graph $G = (V,E)$ on an orientable surface $\Sigma$ of genus $g\geq 1$ so that $\mathcal{B}$ has the same rotation system as a plane embedding $\mathcal{E}$ of $G$.
    Let $f$ be a face of the plane embedding $\mathcal{E}$, and let $T$ be a spanning tree of $G$.
    Then $\mathcal{B}$ can be reconfigured to an embedding $\mathcal{B}'$ in
    an $\varepsilon$-strip system of $\mathcal{B}(T)$ so that the outer face of $\mathcal{B}'$ corresponds to $f$, and the rotation system remains the same throughout the reconfiguration sequence.
\end{lemma}
\begin{proof}
    Let $\varepsilon>0$ be sufficiently small so that
    $\mathcal{B}(T)$ admits an $\varepsilon$-strip system, $U_\varepsilon$ intersects an edge segment of $\mathcal{B}$ if and only if it is part of $\mathcal{B}(T)$ or incident to a vertex in $V$; and its intersection with any such edge segment is connected.
    Recall that $U_\varepsilon$ is a neighborhood of $\mathcal{B}(T)$
    which contains the disks and strips of the $\varepsilon$-strip system.
    In particular, $U_\varepsilon$ is homeomorphic to a disk.
    Similarly, we can define a $\delta$-strip system for $\mathcal{E}(T)$ in the plane, and let $U_\delta$ be a neighborhood of $\mathcal{E}(T)$.
    Since $\mathcal{B}$ and $\mathcal{E}$ have the same rotation system, there is a bijection $\varphi:\partial U_\varepsilon\to \partial U_\delta$
    that maps the crossings $\mathcal{B}(e)\cap \partial U_\varepsilon$ to the crossings $\mathcal{E}(e)\cap \partial U_\delta$ for all $e\in E$.

    We may assume, without loss of generality, that $f$ is the outer face in $\mathcal{E}$. Every edge $e\in E$ that is not in $T$ determines a unique cycle $C_e$ in $T\cup\{e\}$, and its embedding $\mathcal{E}(C_e)$ is a Jordan curve that encloses one or more bounded faces in $\mathcal{E}$. For every $e\in E\setminus E(T)$, let ${\rm rank}(e)$ be the number of bounded faces of $\mathcal{E}$ enclosed by $\mathcal{E}(C_e)$. (Note, however, that $\mathcal{B}(C_e)$ is not necessarily a separating cycle in $\Sigma$.)

    \begin{figure}[htbp]
        \centering
        \includegraphics[width=\textwidth]{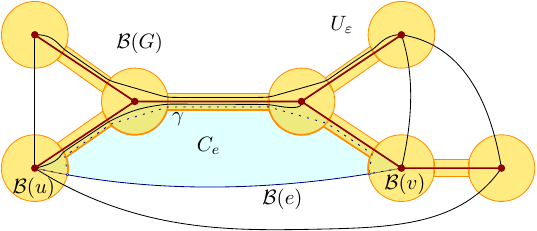}
        \caption{Rerouting an edge $e\in E\setminus E(T)$ to an edge curve in the $\varepsilon$-strip system of $T$.}
        \label{fig:hug-a-tree}
    \end{figure}

    We describe an algorithm that reconfigures $\mathcal{B}$ to an embedding $\mathcal{B}'$ contained in $U_\varepsilon$ such that its outer face w.r.t.\ $U_\varepsilon$ corresponds to face $f$ (of the plane embedding $\mathcal{E}$). We reconfigure every edge $e\in E\setminus E(T)$ in increasing order by rank (ties are broken arbitrarily).

    Assume that we have already rerouted some of the edges, and $e\in E\setminus E(T)$ is the next edge to handle; see \Cref{fig:hug-a-tree}.
    We may assume, without loss of generality, that $\mathcal{B}$ is the current embedding, and $\mathcal{B}$ embeds all edges of rank less than ${\rm rank}(e)$ in the interior of $U_\varepsilon$.
    For every edge $e'\in E\setminus E(T)$, if $\mathcal{E}(e')$
    lies in the interior of $\mathcal{E}(C_e)$, then ${\rm rank}(e')<{\rm rank}(e)$, and by assumption $\mathcal{B}(e')$ is already in the interior of $U_\varepsilon$.
    By the choice of $\varepsilon>0$, $\mathcal{B}(e)$ intersects $\partial U_\varepsilon$ in two points, which partition the closed curve $\partial U_\varepsilon$ into two arcs. Let $\gamma$ be the arc of $\partial U_\varepsilon$ for which $\varphi(\gamma)$ is in the interior of $\mathcal{E}(C_e)$. Note that $\gamma$ cannot cross any edge curve of $\mathcal{B}$, since $\varphi$ would map any crossing $\mathcal{B}(e')\cap \gamma$ to the interior of $\mathcal{E}(C_e)$, which would imply that ${\rm rank}(e')<{\rm rank}(e)$, and $\mathcal{B}(e')$ lies in the interior of $U_\varepsilon$.

    Consequently, we can now reroute $e$ in the $\varepsilon$-strip system closely following the arc $\gamma$ between the two intersection points in $\mathcal{B}(e)\cap \partial U_\varepsilon$. Clearly, by rerouting edge $e$, we did not change the rotation system.

    After successively rerouting all edges in $E\setminus E(T)$, let $\mathcal{B}'$ denote the resulting $\Sigma$-embedding of $G$.
    Note that all bounded faces of $\mathcal{B}'$ lie in the interior of $U_\varepsilon$.
\end{proof}

If $G$ is a tree, we cannot just apply \cref{cor:highergenus} because \cref{cor:highergenus} does not make any assumptions about the rotation system of $\mathcal{B}$ and $\mathcal{R}$.
In fact, some of the steps in the reconfiguration sequence constructed in \Cref{sec:torus} may change the rotation system.
Since the rotation system of an embedding is determined by small neighborhoods of the vertices, we show how to reconfigure $\mathcal{B}$ to $\mathcal{R}$ restricted to such neighborhoods.

\begin{lemma}\label{lem:locally}
    Let $\cal B$ and $\cal R$ be two embeddings of a graph $G=(V,E)$ on an orientable surface $\Sigma$ such that $\cal B$ and $\cal R$ have the same rotation system.
    Then we can reconfigure $\mathcal{B}$ and $\mathcal{R}$ to $\Sigma$-embeddings $\mathcal{B}'$ and $\mathcal{R}'$, resp., so that every vertex $v\in V$ has a neighborhood $N_v$ where $\mathcal{B}'\cap N_v= \mathcal{R}'\cap N_v$.
    Furthermore, the rotation system remains the same throughout the reconfiguration sequence.
\end{lemma}
\begin{proof}
    Since the vertices in $V$ have disjoint neighborhoods, it is enough to prove the claim for one vertex. Let $v\in V$, and let $D_v$ be a small disk in $\Sigma$ centered at $v$ that intersects only segments of the edge curves in $\mathcal{B}$ and $\mathcal{R}$ that are incident to $v$; see~\Cref{fig:locally}. Let $C_v$ be a cone with apex $v$ that is disjoint from all segments incident to $v$ in both $\mathcal{B}$ and $\mathcal{R}$.
    (Note that the counterclockwise first edge incident to $v$ after $C_v$ may be different in $\mathcal{B}$ and $\mathcal{R}$.)

    \begin{figure}[htbp]
        \centering
        \includegraphics[width=\textwidth]{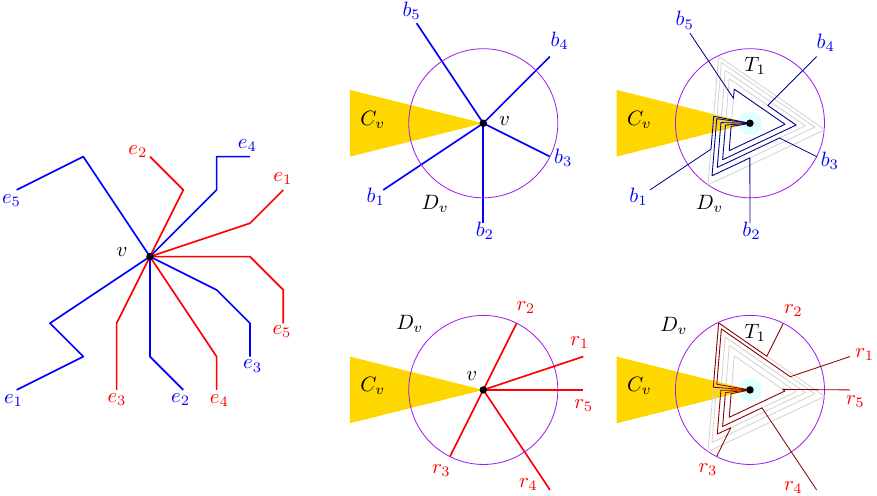}
        \caption{Left: A vertex $v$ of degree $5$ and its incident edge curves in $\mathcal{B}$ (blue) and $\mathcal{R}$ (red).
            Middle: The disk $D_v$, the cone $C_v$, and the incident segments $b_1,\ldots,b_5$ of $\mathcal B$ (top) and $r_1,\ldots,r_5$ of $\mathcal R$ (bottom), where $b_i$ and $r_i$ belong to the same edge.
            Right: After rerouting along nested triangles, the resulting embeddings $\mathcal{B}'$ and $\mathcal{R}'$ agree inside the disk $D_v$.}
        \label{fig:locally}
    \end{figure}

    Let $d=\deg(v)$ be the degree of $v$ in $G$. Create $d$ regular triangles centered at $v$ in the interior of $D_v$, and denote them by $T_1\supset \ldots \supset T_d$.
    Order the $d$ segments of $\mathcal{B}$ incident to $v$ in counterclockwise order starting from the cone $C_v$, as $b_1,\ldots , b_d$. For $i=1,\ldots,d$, redraw the portion of $b_i$ from the point $b_i\cap \partial T_i$ to $v$ as follows: follow $\partial T_i$ clockwise to a point in the interior of the cone $C_v$, and then continue to $v$ in the interior of $C_v$ along a line segment $q_i$.
    Denote by $\mathcal{B}'$ the resulting embedding of $G$.

    We follow the same process to reconfigure $\mathcal{R}$ to $\mathcal{R}'$ with a few changes. Order the $d$ segments of $\mathcal{R}$ incident to $v$ by $r_1,\ldots , r_d$
    such that $r_i$ and $b_i$ belong to the same edge $e_i\in E$ for $i=1,\ldots , d$.
    Let $r_k$ be the first segment encountered counterclockwise after the cone $C_v$.
    For $i=k,\ldots,d$, redraw the portion of $r_i$ from $r_i\cap\partial T_i$ to $v$ by following the curve $\partial T_i$ clockwise to a point in the interior of the cone $C_v$, and reach $v$ along the segment $q_i$. For $i=1,\ldots,k-1$, we follow $\partial T_{k-i}$ \emph{counterclockwise}, and reach $v$ along the segment $q_i$.

    Note that both $\mathcal{B}'(e_i)$ and $\mathcal{R}'(e_i)$ reach $v$ along the same segment $q_i$.
    Consequently, for a sufficiently small disk $N_v$ centered at $v$ inside $T_d$, both embeddings intersect $N_v$ only along the common segments $q_i$, and hence $\mathcal{B}'\cap N_v= \mathcal{R}'\cap N_v$, as required.
\end{proof}

\begin{lemma}\label{lem:reduction-to-tree}
    Let $T = (V,E)$ be a tree and let $\mathcal{B}$ and $\mathcal{R}$ be two
    $\Sigma$-embeddings of $T$ on an orientable surface $\Sigma$ of genus $g\geq 1$
    so that $\mathcal{B}$ and $\mathcal{R}$ have the same rotation system.
    Then $\mathcal{B}$ can be reconfigured to $\mathcal{R}$ so that the rotation system remains the same throughout the reconfiguration sequence.
\end{lemma}
\begin{proof}
    By \Cref{lem:locally}, we can assume that $\mathcal{B}$ and $\mathcal{R}$ not only have the same rotation system, but they are identical in an $\varepsilon$-neighborhood of every vertex $v\in V$, for a sufficiently small $\varepsilon>0$ (see the light-blue shaded disks in \cref{fig:locally}(right)).
    Subdivide each edge $e=uv\in E$ with two new vertices, $u_e$ and $v_e$, and embed them on the edge curve $\mathcal{B}(e)$ (resp., $\mathcal{R}(e)$) in the $\varepsilon$-neighborhood of $u$ and $v$.
    Denote by $T'=(V',E')$ the resulting tree, and by $\mathcal{B}'$ and $\mathcal{R}'$ the two embeddings of $T'$.
    Note that $\mathcal{B}'$ and $\mathcal{R}'$ have the same rotation system.
    Furthermore, all edges incident to the original vertices in $V$ are crossing-free.

    By \Cref{cor:highergenus}, there is a reconfiguration sequence from $\mathcal{B'}$ to $\mathcal{R}'$.
    Specifically, this sequence is constructed in two phases (\Cref{sec:torus}). Phase~1 modifies edge curves in the neighborhoods of crossings with the boundary of the fundamental polygon $\partial \bar{\Sigma}$.
    It does not modify small neighborhoods of the original vertices in $V$, and so it does not change the rotation system.
    Phase~2 considers  edges $e\in E'$ where $\mathcal{B}'(e)\neq \mathcal{R}'(e)$, in a leaf-to-root order after choosing an arbitrary root.
    In each iteration, it modifies edge curves in $\mathcal{B}'$ and $\mathcal{R}'$ in the neighborhoods of red-blue crossings, and ultimately reroutes $\mathcal{B}'(e)$ to $\mathcal{R}'(e)$.
    As such, the rotation system of the original vertices in $V$ remains the same throughout the reconfiguration sequence.
    Note that Phase~2 never modifies edge curves that are already fixed (that is, $\mathcal{B}'(f)=\mathcal{R}'(f)$), and so it only reroutes edges between subdivision vertices.
    A vertex of degree 2 has only one possible rotation, and consequently, Phase~2 also maintains the same rotation system.
\end{proof}

If $G$ is disconnected, it is not obvious how to augment it to a connected planar graph so that the augmented graph has the same rotation system in both embeddings.
We can now show that if all connected components are already in small neighborhoods of their spanning trees, with a consistent outer face, then the augmentation is fairly easy.

\begin{lemma}\label{lem:augmentation}
    Let $\cal B$ and $\cal R$  be two $\Sigma$-embeddings of a disconnected planar graph $G$ on an orientable surface $\Sigma$ of genus $g\geq 1$ so that $\cal B$ and $\cal R$ have the same rotation system, all connected components of $G$ are incident to a common face $\mathcal{B}(f)$ and $\mathcal{R}(f)$, resp., and every connected component of $G$ has the same facial walk in $\mathcal{B}(f)$ and $\mathcal{R}(f)$.
    Then we can augment $G$ to a connected planar graph $G'$, and extend the embeddings $\mathcal{B}$ and $\mathcal{R}$ to embeddings  $\mathcal{B}'$ and $\mathcal{R}'$ of $G'$ such that $\mathcal{B}'$ and $\mathcal{R}'$ have the same rotation system
\end{lemma}
\begin{proof}
    Let $G_1,\ldots,G_h$ denote the connected components of $G$, for $h\geq 2$, and let $v_i\in V(G_i)$ be an arbitrary vertex along the face $f$ for $i=1,\dots,h$.
    We augment $G$ with a star centered at $v_1$ on the vertices $\{v_1,\ldots , v_h\}$.
    We insert the new edges in the embeddings $\mathcal{B}$ and $\mathcal{R}$ as follows.
    Consider the embedding $\mathcal{B}$ (the case of $\mathcal{R}$ is analogous): For $i=2,\ldots, h$,
    embed the edge $v_1v_i$ in the face $\mathcal{B}(f)$ such that the edges $v_1v_2,\ldots, v_1v_h$ are in counterclockwise cyclic order around $v_1$.
    In each iteration, the new edge curve $\mathcal{B}(v_1v_i)$ connects two components of the boundary of $\mathcal{B}(f)$.
    In particular, it does not split $\mathcal{B}(f)$ into two faces, and the vertices $v_1,\dots, v_h$ remain on the boundary of $\mathcal{B}(f)$.
\end{proof}

Finally, if $G$ is connected and both $\mathcal{B}$ and $\mathcal{R}$ are in some small neighborhood of a spanning tree, then we can reconfigure $\mathcal{B}$ to $\mathcal{R}$ using \Cref{lem:reduction-to-tree}, by reconfiguring \enquote{bundles} of edges that are parallel to an edge of a tree.
Formally, the embedding $\mathcal{B}$ will remain in the $\varepsilon$-strip system of $\mathcal{B}(T)$ while we apply a reconfiguration sequence to $\mathcal{B}(T)$.

\begin{proof}[Proof of \Cref{thm:embedding-preserving}]
    By \Cref{lem:augmentation}, we may assume that $G$ is connected, has a spanning tree $T$, and $\mathcal{B}$ and $\mathcal{R}$ lie in the $\varepsilon$-strip systems of $\mathcal{B}(T)$ and $\mathcal{R}(T)$, resp., for a sufficiently small $\varepsilon>0$.

    In particular, $\mathcal{B}(T)$ and $\mathcal{R}(T)$ have the same rotation system. \Cref{lem:reduction-to-tree} gives a reconfiguration sequence
    \[
        (\mathcal{B}(T)=\mathcal{F}_0,\mathcal{F}_1,\ldots , \mathcal{F}_k=\mathcal{R}(T)),
    \]
    such that the embeddings $\mathcal{F}_0,\ldots , \mathcal{F}_k$ have the same rotation system. Observe that the proof of \Cref{lem:reduction-to-tree} (\Cref{sec:torus}) used three operations to construct this reconfiguration sequence:
    \begin{itemize}
        \item subdivide an edge with a new vertex;
        \item suppress a vertex of degree 2;
        \item replace an edge curve $\gamma$ with an edge curve $\gamma'$, with the same endpoints, where $\gamma$ and $\gamma'$ intersect only at their endpoints.
    \end{itemize}
    We follow this sequence and create a sequence of $\Sigma$-embeddings of $G$ in the $\varepsilon$-strip system of $\mathcal{F}_i$.
    Furthermore, we maintain the property that if $v$ is a subdivision vertex created during the algorithm, incident to edges $vu,vw\in E(T)$,
    then for every edge $e\in E\setminus E(T)$, every connected component of $D_u\cap \mathcal{F}_i(e)$ connects $A_{v,u}$ and $A_{v,w}$.
    Assume that $\varepsilon$ is sufficiently small such that $\mathcal{F}_i$ admits an $\varepsilon$-strip system for all $i=0,1,\ldots ,k$.

    Consider the $i$-th step of the reconfiguration sequence of $T$, from $\mathcal{F}_{i-1}$ to $\mathcal{F}_i$.
    If an edge $uw\in E(T)$ is subdivided
    with a vertex $v$, we can easily update the $\varepsilon$-strip system: create a disk $D_v$ and replace the strip $S_{uw}$ with $S_{uv}$ and $S_{vw}$.
    The embedding of $G$ does not change.
    Similarly, if a vertex $v$ of degree $2$ is suppressed, the above additional property guarantees that the embedding of $G$ remains the same.

    Assume that an edge curve $\gamma=\mathcal{F}_{i-1}(e)$ is replaced by $\gamma'=\mathcal{F}_i(e)$ for some edge $e=uv\in E(T)$.
    Recall that $\mathcal{F}_{i-1}$ and  $\mathcal{F}_i$ have the same rotation system, so $\gamma$ and $\gamma'$ have the same cyclic order with respect to all other edges incident to $u$ and $v$.
    However, the cyclic order between $\gamma$ and $\gamma'$ is either the same or reversed at $u$ and $v$; see \Cref{fig:rotations1}.
    Let $S_{u,v}$ and $S'_{u,v}$ be the strips corresponding to $\gamma$ and $\gamma'$, resp., in the strip system.

    \begin{figure}[htbp]
        \centering
        \includegraphics[width=\textwidth]{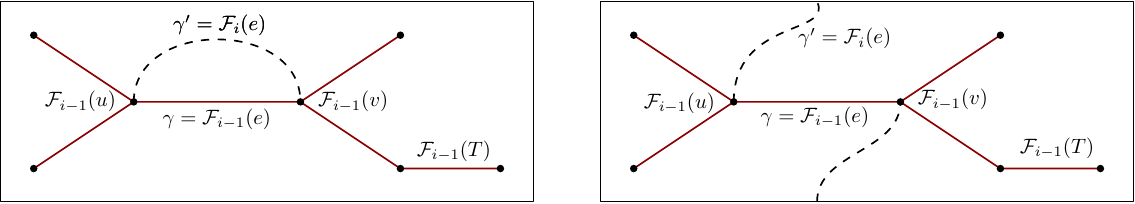}
        \caption{The edge curves $\gamma$ and $\gamma'$ may be in the same (right) or reverse (left) cyclic order at $u$ and $v$.}
        \label{fig:rotations1}
    \end{figure}

    Let $S_{u,v}$ and $S'_{u,v}$ be the strips corresponding to $\gamma$ and $\gamma'$, resp., in the strip system.
    Assume w.l.o.g.\ that $(\gamma,\gamma')$ is in counterclockwise order on $\partial D_u$.
    Let $\gamma_1,\ldots, \gamma_m$ be the edge curves that traverse $S_{u,v}$, in clockwise order w.r.t.\ the arc $A_{u,v}\subset \partial D_u$.
    We can successively reroute the portions of these edge curves in $D_u\cup S_{u,v}\cup D_v$, in this order, so that they pass through $S'_{u,v}$; see \Cref{fig:rotations2}(left).

    \begin{figure}[htbp]
        \centering
        \includegraphics[width=\textwidth]{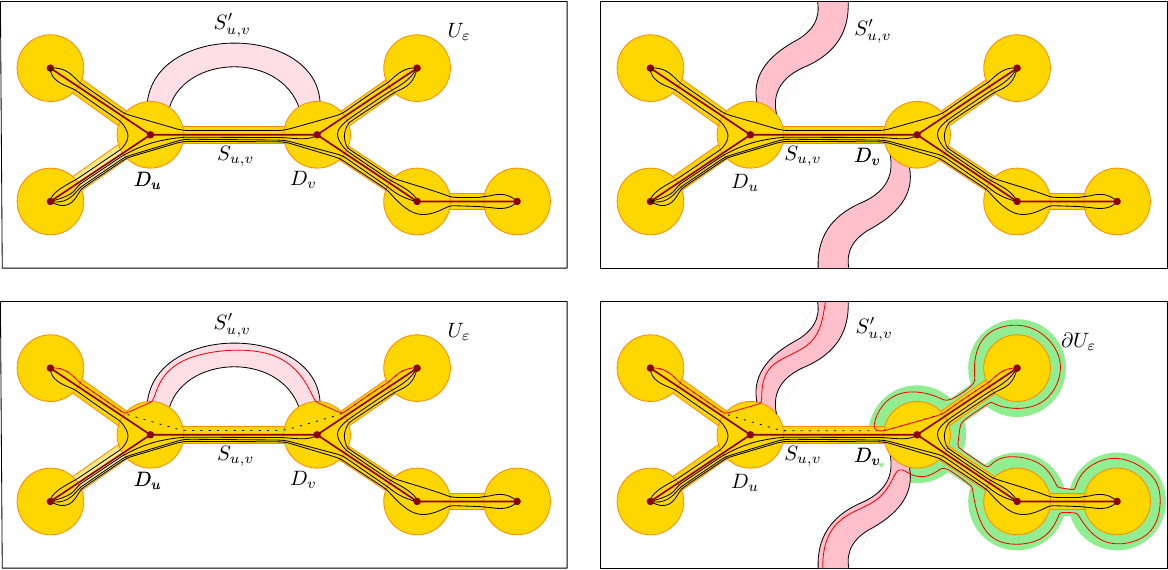}
        \caption{Rerouting the edge curves from the strip $S_{u,v}$ to the strip $S'_{u,v}$.}
        \label{fig:rotations2}
    \end{figure}

    If $\gamma$ and $\gamma'$ have reverse rotations at $u$ and $v$, then
    $\gamma_1,\ldots ,\gamma_m$ cross the arc $A_{v,u}\subset \partial D_v$ in counterclockwise order.
    We can successively reroute $\gamma_1,\ldots ,\gamma_m$ in $D_u\cup S_{u,v}\cup D_v$.

    If $\gamma$ and $\gamma'$ have the same rotations at $u$ and $v$, then
    $\gamma_1,\ldots ,\gamma_m$ reach the arc $A_{v,u}\subset \partial D_v$ in clockwise order.
    For $j=1,\ldots ,m$, we reroute $\gamma_j$ to a curve $\gamma_j'$ that goes through the strip $S'_{u,v}$ to $\partial D_v$, and then follows the boundary $\partial U_\varepsilon$ counterclockwise until it reaches $\partial D_v$ again.
    After traversing the strip $S'_{u,v}$, the edge curves $\gamma_1',\ldots, \gamma_m'$ reach $\partial D_v$ in counterclockwise order.
    A detour around $\partial U_\varepsilon$ reverses their order, and $\gamma_1',\ldots, \gamma_m'$ ultimately reach $\partial D_v$ in clockwise order, as required.
    This concludes the description of the rerouting step.

    For each reconfiguration step of $T$, from $\mathcal{F}_{i-1}$ to $\mathcal{F}_i$, we can compute a reconfiguration sequence of $G$ from a $\varepsilon$-strip system of $\mathcal{F}_{i-1}$ to that of $\mathcal{F}_i$.
    In intermediate steps, we maintain an embedding of $G$ into an $\varepsilon$-strip system of $\mathcal{F}_{i-1}\cup \mathcal{F}_i$.
    Since the curves $\gamma=\mathcal{F}_{i-1}(e)$ and $\gamma'=\mathcal{F}_i(e)$ do not cross each other, this is a $\Sigma$-embedding of $G$.
    Overall, our algorithm computes a reconfiguration sequence of $G$ from $\mathcal{B}$ to $\mathcal{R}$.
\end{proof}


\section{Omitted Material from \cref{ssec:series-parallel}}
\label{app:series-parallel}

\thmseriesparallel*

\begin{proof}
    \label{Pseriesparallel}
    We describe the reconfiguration for $g=1$.
    The case $g>1$ can be reduced to $g=1$ by ignoring all but two sides of $\partial\bar\Sigma$, which we call the \emph{horizontal} and \emph{vertical part} of the boundary $\partial\bar\Sigma$.
    Note that the names correspond to their actual orientation in the case $g=1$.

    If for each pair of poles, the sorting of the parallel subgraphs is identical in $\mathcal{B}$ and $\mathcal{R}$, then $\mathcal{B}$ and $\mathcal{R}$ have the same rotation system,
    and the reconfiguration is possible by \cref{thm:embedding-preserving}.

    Otherwise, let $s$ and $t$ be a pair of poles such that their parallel subgraphs are not sorted the same in $\mathcal{B}$ and $\mathcal{R}$.
    More precisely, let $G_1,\ldots, G_k$ denote the set of maximal subgraphs separated by the separation pair $\{s,t\}$ in $G$.
    Assume that the edges of $G_1,\ldots, G_k$ incident to $s$ and $t$ occur in this order clockwise around $s$ and in this order counterclockwise around $t$ in $\cal B$, respectively.
    In particular, also observe that $G_i$ may have several edges incident to $s$ or $t$, but those occur consecutively in the list of edges sorted according to the clockwise cyclic order around $s$.
    Note that this is necessarily the case for a plane embedding (up to renaming of the components).
    Moreover, let $G_i$ and $G_{i+1}$ be such that their order is interchanged in $\cal R$, i.e., in $\cal R$, the edges of $G_{i+1}$ precede those of $G_i$ in a clockwise (resp., counterclockwise) walk around $\cal B$ (resp., $\cal R$).

    Starting from $\mathcal{B}$, we now first use \cref{thm:embedding-preserving} to make the face between $G_i$ and $G_{i+1}$ the face containing the boundary $\partial\bar{\Sigma}$.
    This yields an embedding $\mathcal{B}'$; see also the first subfigure in \cref{fig:seriesparallel}.
    Let now $(s_1,\ldots,s_k)$ denote the edges incident to $s$ in $G_{i}$ sorted such that $s_x$ is on the outer face (containing the boundary $\partial\bar{\Sigma}$) if all of $s_1,\ldots,s_{x-1}$ get deleted from the embedding $\mathcal{B}'$.
    We reroute $s_1,\ldots,s_k$ in order of increasing index such that we split the currently processed edge $s_x$ at an interior point of its edge curve $\mathcal{B}'(s_x)$ and reroute both ends to the horizontal part of the boundary $\partial\bar\Sigma$; see also second subfigure in \cref{fig:seriesparallel}.
    Note that after these reroutings the vertical part of the boundary $\partial\bar{\Sigma}$ is entirely accessible both from the outside boundary of $G_{i+1}$ but also from the inside boundary of $G_i$ (w.r.t.\ to $\mathcal{B}'$).

    Next, we reroute the edges incident to $s$ and $t$ from component $G_{i+1}$.
    To this end let $(s'_1,\ldots,s'_j)$ and $(t'_1,\ldots,t'_\ell)$ denote the edges incident to $s$ and $t$ respectively in $\mathcal{B}'$ sorted such that $s'_x$ or $t'_x$ is on the face containing the boundary $\partial\bar\Sigma$ if all of $s'_1,\ldots,s'_{x-1}$ or $t'_1,\ldots,t'_{x-1}$, respectively, get deleted from the embedding $\mathcal{B}'$.
    We now first reroute all of $s'_1$ to $s'_j$ in increasing order of their indices.
    As with the $s_x$ in the previous step, we cut $\mathcal{B}'(s'_x)$.
    Then, we redraw the end incident to $s$ so as to occur in the cyclic order around $s$ in between $G_i$ and $G_{i-1}$ and from there to touch the vertical boundary, whereas we reroute the other end to the vertical boundary as well.
    Afterward, we deal analogously with $t'_1$ to $t'_\ell$.
    As a result, we have now changed the cyclic order of $G_i$ and $G_{i+1}$ around $s$ and $t$ successfully; see also the third subfigure in \cref{fig:seriesparallel}.
    Moreover, we observe that there is now a plane path $\pi_s$ starting at $s$, traversing $G_{i+1}$'s boundary to $t$ and then the boundary of $G_i$ until reaching the vertex of $s_1$ other than $s$.
    More generally, if we remove the edge curves of $s_1, \ldots, s_{x-1}$, then $\pi_s$ can be extended so as to contain the vertex of $s_x$ other than $s$.

    We use this property to next reroute $s_1$ to $s_k$ again, once more in increasing order of the indices.
    In each iteration, we reroute $s_x$ closely following $\pi_s$, which makes $s_x$ avoid $\partial \bar{\Sigma}$ again; see the fourth subfigure in \cref{fig:seriesparallel}.
    After performing all of these reroutings, there is a plane path $\pi_{s'}$ between $s$ and the vertex of $s'_1$ other than $s$.
    More generally, if we remove the edge curves of $s'_1,\ldots,s'_{x-1}$, then $\pi_{s'}$ can be extended so as to contain the vertex of $s'_x$ other than $s$.
    Similarly, there is a plane path $\pi_{t'}$ between $t$ and the vertex of $t'_1$ other than $t$.
    Again, if we remove the edge curves of $t'_1,\ldots,t'_{x-1}$, then $\pi_{t'}$ can be extended so as to contain the vertex of $t'_x$ other than $t$.

    Thus, we can reroute $s'_1$ to $s'_j$ and $t'_1$ to $t'_\ell$ in increasing order of their indices, closely tracing $\pi_{s'}$ and $\pi_{t'}$, respectively, similar to the previous step where we rerouted $s_1$ to $s_k$; see the fifth subfigure in \cref{fig:seriesparallel}.
    Since $\pi_{s'}$ and $\pi_{t'}$ are plane, we now obtain once more a plane embedding where $G_{i+1}$ precedes $G_i$ in clockwise order around $s$ (in counterclockwise order around $t$, respectively).

    Repeated application of the subroutine that swaps the order of $G_i$ and $G_{i+1}$ can be used to sort the parallel components at a single pair of poles $s$ and $t$ (e.g., mimicking the bubble sort algorithm).
    As other pairs of poles can be treated similarly, the proof then follows by application of \cref{thm:embedding-preserving}.
\end{proof}

\section{Omitted Material from \cref{sec:projective}}
\label{app:projectivePlain}

\claimprojectiveplane*

\begin{claimproof}
    \label{Pclaimprojectiveplane}
    Recall that $N_i$ is the thickening of $\partial D\cup M_i\cup \left(\bigcup_{j< i}M_j\right)$; arcs $\gamma_1,\ldots ,\gamma_t$ intersect $M_i\setminus \{s_i\}$ (ordered from $t_i$ to $s_i$); and $P(\gamma_j)$ is the curve in $\mathcal{B}$ that contains the arc $\gamma_j$ for $j=1,\ldots , t$. Let $k$ be the smallest integer such that $P(\gamma_k)=P(\gamma_j)$ for some $j<k$.
    In particular, this means that $P(\gamma_k)$ intersects $M_i\setminus \{s_i\}$ at least twice.

    The two intersection points, $\gamma_j\cap (M_i\setminus \{s_i\})$ and $\gamma_k\cap (M_i\setminus \{s_i\})$, are connected by two arcs: along $M_i$ and along $P(\gamma_j)=P(\gamma_k)$; the union of these two arcs is a Jordan curve that we denote by $C$.
    Refer to \cref{fig:simplify-a}.

    \begin{figure}[htbp]
        \centering
        \includegraphics[width=0.6\textwidth]{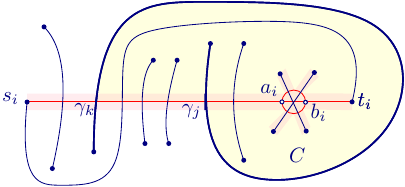}
        \caption{$P(\gamma_j)=P(\gamma_k)$ crosses $M_i$ twice.}
        \label{fig:simplify-a}
    \end{figure}

    Note that the closed curve $C$ cannot separate $s_i$ from $t_i$.
    Indeed, suppose for contradiction that $s_i$ and $t_i$ are on opposite sides of $C$ (as in  \cref{fig:simplify-a}). Since the curve $P_i$ connects $s_i$ and $t_i$, then $P_i$ must cross $C$.
    Since the curves in $\mathcal{B}$ are disjoint, then $P_i$ must cross $C\cap M_i$, and so $P_i$ intersects $M_i\setminus \{s_i\}$ twice.
    This contradicts the minimality of $k$.
    In the remainder of the proof, we may assume that $C$ does not separate $s_i$ and $t_i$.
    We distinguish between two cases.

    \subparagraph{Case~1: Every curve $\bm{P(\gamma_\ell)}$, $\bm{j<\ell<k}$, crosses $\bm{M_i\setminus \{s_i\}}$ only once}
    We can redraw $P(\gamma_j)=P(\gamma_k)$, by replacing its subcurve between (and including) $\gamma_j$ and $\gamma_k$, such that the new arc does not cross $a_i s_i$.
    The replacement arc closely follows $a_i s_i$ in the interior of $C$, and makes a detour along each arc $P(\gamma_\ell)$, $j<\ell<k$, in the interior of $C$.
    See \cref{fig:simplify-b} (specifically, \cref{fig:simplify-b}~(top) shows an example for $P(\gamma_k)\neq P_i$, and \cref{fig:simplify-b}~(bottom) for $P(\gamma_k)=P_i$).

    \begin{figure}[htbp]
        \centering
        \includegraphics[width=\textwidth]{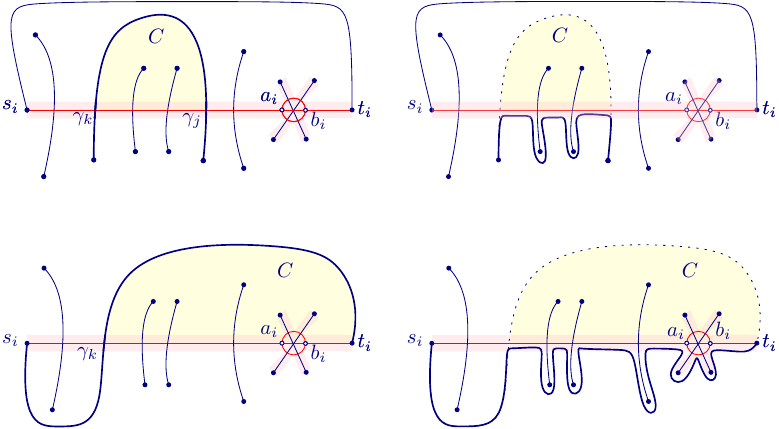}
        \caption{Every curve $P(\gamma_\ell)$, $j<\ell<k$, crosses $M_i\setminus \{s_i\}$ only once.}
        \label{fig:simplify-b}
    \end{figure}

    \subparagraph{Case~2: Some curve $\bm{P(\gamma_h)}$, $\bm{j<h<k}$, crosses $M_i\setminus \{s_i\}$ more than once}
    We shall use the crosscap, and modify the curves $\{P(\gamma_1),\ldots , P(\gamma_{k-1})\}$ in five stages; refer to \cref{fig:simplify-c}.
    Orient $M_i$ such that $C$ is on its \emph{left} side; and then orient the curves in $\{P(\gamma_1),\ldots , P(\gamma_{k-1})\}$ such that they each start from the right side of $M_i$, and cross to the left side of $M_i$.

    \begin{figure}[t]
        \centering
        \includegraphics[width=\textwidth]{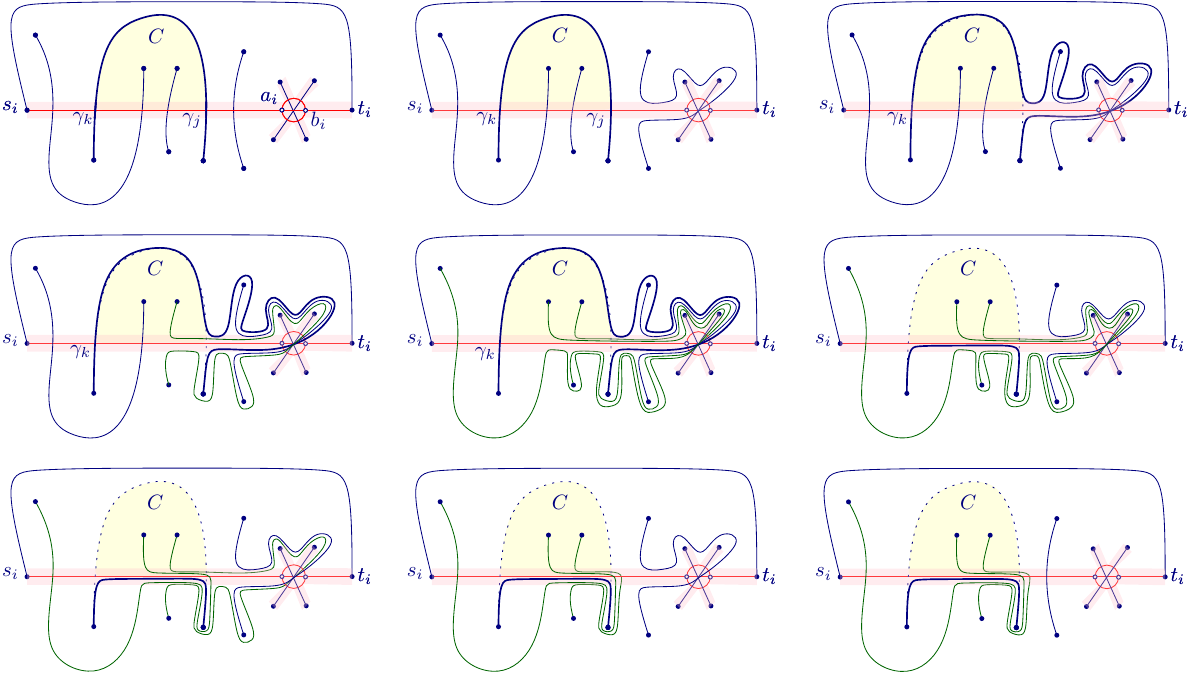}
        \caption{Curve $P(\gamma_h)$, $j<h<k$, crosses $M_i\setminus \{s_i\}$ more than once.}
        \label{fig:simplify-c}
    \end{figure}

    \begin{itemize}
        \item Step~1: For $\ell=1, \ldots, j$, modify $P(\gamma_\ell)$ to replace the arc $\gamma_\ell$ with a new arc $\gamma'_\ell$ that closely follows $M_i$ to $\partial D_i$, crosses the crosscap, and then closely follows the curves incident to $\partial D_i$.
        \item Step~2: For $\ell=j+1,\ldots k-1$, replace $\gamma_\ell$ with an arc $\gamma'_\ell$ that closely follows the union of $M_i$ and the arcs redrawn in Stage~1, crosses the crosscap, and then closely follows $M_i$.
        \item Step~3: Replace $P(\gamma_j)=P(\gamma_k)$ such that its arc between $M_i\cap \gamma_j$ and $M_i\cap \gamma_k$ is replaced by a new arc that closely follows $M_i$.
              As a result  $P(\gamma_j)=P(\gamma_k)$ has two fewer crossings with $M_i$.
        \item Step~4: For $\ell=k-1,\ldots, j+1$, replace $\gamma'_\ell$ with a new arc $\gamma''_\ell$ that closely follows $P(\gamma_j)=P(\gamma_k)$, and crosses $M_i$ exactly once.
        \item Step~5: For $\ell=j-1,\ldots , 1$, replace $\gamma'_\ell$ with the original arc $\gamma_\ell$ (which crosses $M_i$ exactly once).
    \end{itemize}
    The number of crossings between $\{P_i,\ldots ,P_n\}$ and $M_i\setminus \{s_i\}$ may increase in intermediate steps.
    However, at the end of the process, $P(\gamma_j)=P(\gamma_k)$ has two fewer intersections with $M_i\setminus \{s_i\}$, and each of the other curves has the same number of intersections as before.
\end{claimproof}

\section{Omitted Material from \cref{sec:counter}}
\label{app:negative}

\thmnegative*
\label{Pnegative}

\begin{proof}
    For the plane (and the sphere), Ito et al.~\cite{ItoIK0MNOO25} showed that a perfect matching on 4 vertices $G_0$ has two embeddings, $\mathcal{B}_0$ and $\mathcal{R}_0$, that are not reconfigurable; see also \cref{fig:intro:1}.

    Let $\Delta (abc)$ be a triangle in the plane that contains both $\mathcal{B}_0$ and $\mathcal{R}_0$.
    We construct a graph $G_1$ on 7 vertices as a disjoint union of a 3-cycle $abc$, a perfect matching on four vertices, and 5 additional edges, as shown in \cref{fig:impossible}.
    We also construct two embeddings of $G_1$, denoted $\mathcal{B}_1$ and $\mathcal{R}_1$, that extend $\mathcal{B}_0$ and $\mathcal{R}_0$, respectively; see \cref{fig:impossible}.
    Note that $\mathcal{B}_1$ and $\mathcal{R}_1$ have the same rotation system in the plane.

    \begin{figure}[htbp]
        \centering
        \includegraphics[width=0.8\textwidth]{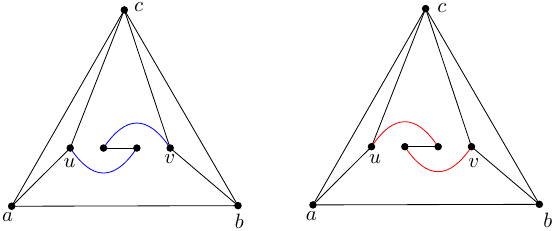}
        \caption{Two embeddings of the graph $G_1$ in the plane.}
        \label{fig:impossible}
    \end{figure}

    Let $\Sigma$ be a surface.
    We construct a finite simple graph $H=H(\Sigma)$ that is uniquely embeddable
    in $\Sigma$ (that is, any two $\Sigma$-embeddings differ only by a homeomorphism of $\Sigma$ and a graph automorphism)
    and whose $3$-cycles are facial in every $\Sigma$-embedding.

    We first choose a finite \emph{simplicial triangulation} $K$ of $\Sigma$, that is, an embedded triangulated graph of $\Sigma$ where every face is homeomorphic to a disk.
    Such a triangulation exists for every compact surface; see, e.g., \cite[Section~3.1]{MT2001}.
    Next, viewing $K$ as its embedded $1$-skeleton, we subdivide each edge once, obtaining a hexagonalization $K'$ whose faces are induced $6$-cycles.
    To complete the construction of $H$, we add a new \emph{face vertex} $v_f$ in the interior of each face $f$ of $K'$ and connect it crossing-free to all vertices on $f$. Observe that $H$ is a triangulation.

    By Negami's Theorem~3.3, a triangulation of a closed surface is uniquely embeddable
    if it has no core skew vertex~\cite[Theorem~3.3, pp.~75--76]{Negami1985}.
    Negami observes that if a triangulation contains a skew vertex,
    then there is a triangle that does not bound a disk in $\Sigma$.
    It therefore suffices to show that every $3$-cycle of $H$ bounds a triangular face.
    To this end, observe that since $K'$ was obtained by subdividing all edges of the triangulation $K$, it does not contain any $3$-cycle.
    Hence, each $3$-cycle contains a face vertex, and since no two face vertices are adjacent, it contains exactly one.
    Write such a cycle as $v_f xy$.
    Both $x$ and $y$ lie on the boundary of $f$.
    Since the boundary of $f$ is an induced $6$-cycle, $xy$ is a boundary edge of $f$, so $v_f xy$ bounds one of the triangular faces added inside $f$.

    Consequently, every $3$-cycle of $H$ bounds a triangular face.
    Hence $H$ has no skew vertex and thus no core skew vertex, so Negami's theorem implies that $H$ is uniquely embeddable in $\Sigma$.
    Since graph automorphisms map $3$-cycles to $3$-cycles, unique embeddability then implies that every $3$-cycle of $H$ is facial in every $\Sigma$-embedding of $H$.

    \begin{claimproof}[Proof of Statement~\ref{thm:negative:1}.] 
        Now, we can construct the graph $H_1=H_1(\Sigma)$ by combining $H$ and $G_1$, namely, by identifying an arbitrary facial cycle of $H$ with the 3-cycle $abc$ in $G_1$.
        We obtain two $\Sigma$-embeddings of $H_1$, denoted $\mathcal{B}$ and $\mathcal{R}$, by respectively inserting the two embeddings of $G_1$ shown in \cref{fig:impossible} into the corresponding face of the chosen embedding of $H$.
        Clearly, $\mathcal{B}$ and $\mathcal{R}$ have the same rotation system on $\Sigma$.

        Suppose that there is a reconfiguration sequence $(\mathcal{E}_0,\ldots ,\mathcal{E}_k)$ such that $\mathcal{E}_0=\mathcal{B}$ and $\mathcal{E}_k=\mathcal{R}$.
        The restriction of each intermediate $\Sigma$-embedding of $H_1$ to $H$ is a $\Sigma$-embedding of $H$.
        By the property established above, the $3$-cycle $abc$ bounds a triangular face $F^i_{abc}$ in the restriction of $\mathcal E_i$ to $H$.
        Since $u$ and $v$ have three independent paths to $a$, $b$, and $c$, the vertices $u$ and $v$ lie in the interior of the face $F^i_{abc}$.
        Consequently, all vertices of $V(G_1)\setminus \{a,b,c\}$ lie in the interior of $F^i_{abc}$.

        We simulate the reconfiguration sequence $(\mathcal{E}_0,\ldots ,\mathcal{E}_k)$ on $\Sigma$ with a reconfiguration sequence $(\mathcal{M}_0,\ldots ,\mathcal{M}_k)$ from $\mathcal{B}_1$ to $\mathcal{R}_1$ in the plane such that the embedding of the 3-cycle $abc$ remains a regular triangle $\Delta (abc)$ at all times.
        We also maintain a homeomorphism $\varphi_i:\mathrm{cl}(F^i_{abc})\to\Delta(abc)$ and the property that for every edge $e\in E(G_1)$, we have $\mathcal{M}_i(e)=\varphi_i\circ \mathcal{E}_i(e)$.
        We simulate a move from $\mathcal{E}_i$ to $\mathcal{E}_{i+1}$ as follows.
        If the move redraws an edge in $E(H_1)\setminus E(G_1)$, we do nothing, that is, $\mathcal{M}_{i+1}=\mathcal{M}_i$ and $\varphi_{i+1}=\varphi_i$.
        If the move redraws an edge $e$ of the 3-cycle $abc$, then the face $F^i_{abc}$ changes, and we update our homeomorphism $\varphi_i$ piecewise as follows; see \Cref{fig:homomorphism}~(left--middle): Let $f_i$ be the face of $\mathcal{E}_i(H_1)$ adjacent to $e$ such that $f_i\subset F^i_{abc}$, and let $f_{i+1}$ be the corresponding face of $\mathcal{E}_{i+1}(H_1)$.
        As both $f_i$ and $f_{i+1}$ are homeomorphic to a disk, there is a homeomorphism $\lambda:\mathrm{cl}(f_i)\to \mathrm{cl}(f_{i+1})$ that fixes every point in $\partial f_i\cap \partial f_{i+1}$. For all points $p\in F^i_{abc}\setminus f_i$, let $\varphi_{i+1}(p)=\varphi_i(p)$; and for all $p\in f_{i+1}$, let $\varphi_{i+1}(p)=\varphi_i(\lambda^{-1}(p))$. We also set $\mathcal{M}_{i+1}=\mathcal{M}_i$.
        Finally, if the move changes any interior edge $e$ of $G_1$ (as in \Cref{fig:homomorphism}~(middle--right)), then the homeomorphism does not change (that is, $\varphi_{i+1}=\varphi_i$), and we replace $\mathcal{M}_i(e)=\varphi_i\circ \mathcal{E}_i(e)$ with $\mathcal{M}_{i+1}(e)=\varphi_{i+1}\circ \mathcal{E}_{i+1}(e)$.

        \begin{figure}[htbp]
        \centering
        \includegraphics[width=0.95\textwidth]{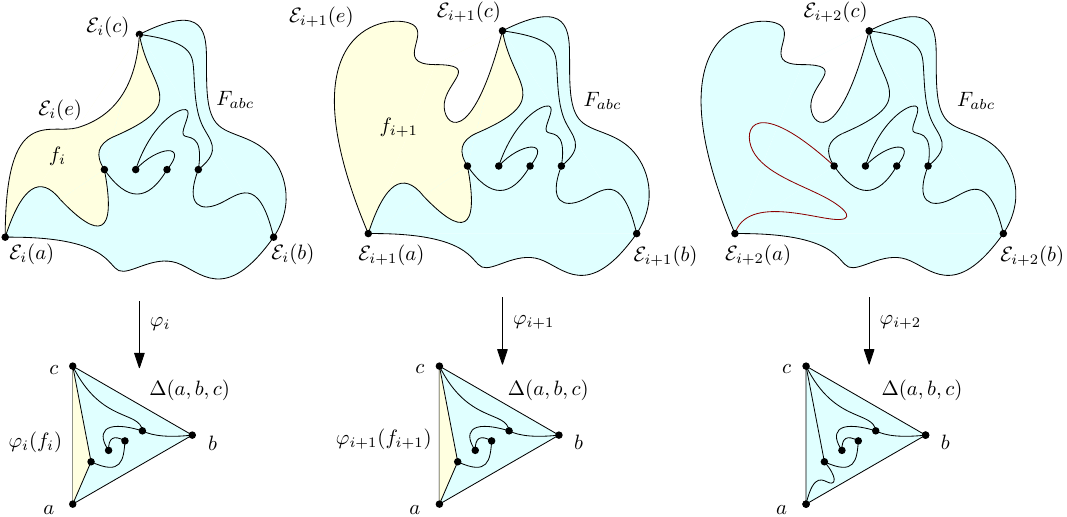}
        \caption{Embeddings $\mathcal{E}_i,\mathcal{E}_{i+1},\mathcal{E}_{i+2}$ in a reconfiguration sequence, and the corresponding homeomorphisms $\varphi_i,\varphi_{i+1},\varphi_{i+2}$.}
        \label{fig:homomorphism}
    \end{figure}

        The reconfiguration sequence $(\mathcal{M}_0,\ldots ,\mathcal{M}_k)$, restricted to the perfect matching $G_0$, yields a reconfiguration sequence from $\mathcal{B}_0$ to $\mathcal{R}_0$ in the plane.
        This contradicts the result by Ito et al.~\cite{ItoIK0MNOO25}.
        Therefore, $\mathcal{B}$ and $\mathcal{R}$ are also not reconfigurable.
    \end{claimproof}

    \begin{claimproof}[Proof of Statement~\ref{thm:negative:2}.]
        We augment the graph $G_1$ from the proof of Statement~\ref{thm:negative:1}, yielding a graph $G_2$.
        Namely, we add two additional edges $uu'$ and $vv'$ between the endpoints of the two matching edges $uv'$ and $u'v$; see the dashed edges in \cref{fig:negative}.
        As in the previous proof, we construct a graph $ H_2 = H_2(\Sigma)$ by identifying the $3$-cycle $abc$ of $G_2$ with an arbitrary facial cycle of the uniquely embeddable graph $H$.

        We now show that fixing the rotation systems $B_2$ and $R_2$ forces the required embeddings of the gadget $G_2$.
        In every $\Sigma$-embedding of $H$, the $3$-cycle $abc$ bounds a triangular face.
        As in the proof of Statement~\ref{thm:negative:1}, the three independent paths from each of $u$ and $v$ to $a$, $b$, and $c$ force the part of $G_2$ outside the cycle $abc$ to lie within this face, which is homeomorphic to a disk.
        Thus, $G_2$ is necessarily embedded as a plane graph.
        For plane graphs, fixing the rotation system actually determines the planar embedding up to the choice of the outer face.
        However, the outer face is prescribed by the fact that we identified $abc$ with a facial cycle of $H$.
        Hence, if we prescribe two rotation systems as shown in \cref{fig:negative}, then two different embeddings for $G_2$ must be used in any $\Sigma$-embedding of $H_2$ implementing $B_2$ and $R_2$, respectively.
        Applying the simulation from the proof of Statement~\ref{thm:negative:1} to a hypothetical reconfiguration sequence between any such pair would yield a plane reconfiguration sequence whose restriction to the edges $uv'$ and $u'v$ transforms $\mathcal B_0$ into $\mathcal R_0$.
        This contradicts the result of Ito et al.~\cite{ItoIK0MNOO25}.
    \end{claimproof}
    We have shown both statements of the theorem, which concludes the proof.
\end{proof}

\end{document}